\documentclass[11pt]{article}
\usepackage[titletoc,title]{appendix}
\usepackage{amsmath,amsfonts,amssymb,bm,amsthm}
\usepackage[colorlinks,citecolor=blue,linkcolor=blue,pagebackref]{hyperref}
\usepackage{tikz}
\usetikzlibrary{shapes,arrows}
\usetikzlibrary{intersections,shapes.arrows}
\usetikzlibrary{decorations.pathreplacing}
\usetikzlibrary{arrows.meta}
\numberwithin{equation}{section}
\newtheorem{theorem}{Theorem}[section]
\newtheorem{lemma}[theorem]{Lemma}
\newtheorem{proposition}[theorem]{Proposition}
\newtheorem{corollary}[theorem]{Corollary}

\theoremstyle{definition}
\newtheorem{example}[theorem]{Example}
\newtheorem{definition}[theorem]{Definition}

\theoremstyle{remark}
\newtheorem{remark}[theorem]{Remark}

\DeclareFontFamily{U}{mathx}{}
\DeclareFontShape{U}{mathx}{m}{n}{<-> mathx10}{}
\DeclareSymbolFont{mathx}{U}{mathx}{m}{n}
\DeclareMathAccent{\widecheck}{0}{mathx}{"71}

\usepackage{multicol}
\usepackage{longtable}
\usepackage{slashed}
\usepackage{accents}
\usepackage[T1]{fontenc}
\usepackage{pifont}
\usepackage[top=1in, bottom=1in, left=1in, right=1in]{geometry}
\usepackage{fancyhdr}
\definecolor{amber}{rgb}{1.0, 0.49, 0.0}

\newcommand{\bigslant}[2]{{\raisebox{.2em}{$#1$}\left/\raisebox{-.2em}{$#2$}\right.}}

\def\mcM{{\mathcal{M}}}
\def\mcL{{\mathcal{L}}}
\def\bC{{\mathbb{C}}}
\def\bH{{\mathbb{H}}}
\def\bR{{\mathbb{R}}}
\usepackage{enumerate}
\usepackage{xcolor,cancel}
\usepackage{relsize}

\renewcommand{\tilde}{\widetilde}

\usepackage{orcidlink}

\allowdisplaybreaks

\usepackage{pifont}
\renewcommand*{\backrefalt}[4]{%
\ifcase #1 %
No citations%
\or
\ding{43}~p.~#2%
\else
\ding{43}~pp.~#2%
\fi}
\usepackage{soul}

\begin{document}

\title{The Hadamard parametrix on half-Minkowski with Robin boundary conditions: Fundamental solutions and Hadamard states \\
}
\author{
Beatrice Costeri\thanks{
Dipartimento di Fisica ``Alessandro Volta'',
Universit\`a degli Studi di Pavia \& INFN, Sezione di Pavia \& INdAM, Sezione di Pavia, 
Via Bassi 6, 
I-27100 Pavia, 
Italia;
beatrice.costeri01@universitadipavia.it.
}
\and
Claudio Dappiaggi\,\orcidlink{0000-0002-3315-1273}\thanks{Dipartimento di Fisica ``Alessandro Volta'',
Universit\`a degli Studi di Pavia \& INFN, Sezione di Pavia \& INdAM, Sezione di Pavia, 
Via Bassi 6, 
I-27100 Pavia, 
Italia;
claudio.dappiaggi@unipv.it,
\url{https://claudiodappiaggi.com}} 
\and 
Benito A. Ju\'arez-Aubry\,\orcidlink{https://orcid.org/0000-0002-7739-4293}
\thanks{Department of Mathematics, University of York, Heslington, York YO10 5DD, UK; benito.juarezaubry@york.ac.uk. \url{https://bajuarezaubry.wordpress.com/}}
\and 
Raman Deep Singh\thanks{\text{Dipartimento di Fisica ``Alessandro Volta'',
Universit\`a di Pavia,
Via Bassi 6, 
I-27100 Pavia, 
Italia;} \text{ramandeep.singh01@universitadipavia.it}}
}


\date{}

\maketitle

\vspace{-.6cm}

\begin{abstract}



We address the problem of constructing fundamental solutions and Hadamard states for a Klein-Gordon field in half-Minkowski spacetime with Robin boundary conditions in $d \geq 2$ spacetime dimensions. First, using a generalisation of the Robin-to-Dirichlet map exploited by Bondurant and Fulling [J. Phys. A: Math. Theor. {\bf 38} 7 (2005)] in dimension $2$, we obtain a representation for the advanced and retarded Green operators in terms of a convolution with the kernel of the inverse Robin-to-Dirichlet map. This allows us to prove the uniqueness and support properties of the Green operators. Second, we obtain a local representation for the Hadamard parametrix that provides the correct local definition of Hadamard states in $d \geq 2$ dimensions, capturing `reflected' singularities from the spacetime boundary. We show that our fundamental solutions abide by this local parametrix representation. Finally, we prove the equivalence of our local Hadamard condition and the global Hadamard condition with a wave-front set described in terms of generalized broken bi-characteristics, obtaining a Radzikowski-like theorem in half-Minkowski spacetime.
\

{\bf Keywords:} Hadamard states, Robin boundary conditions, microlocal analysis

\

{\bf 2020 MSC classes: } 81T20, 35E05

\end{abstract}

\tableofcontents

\allowdisplaybreaks

\section{Introduction}\label{Sec: Introduction} 
Boundary value problems for hyperbolic operators on spacetimes with timelike boundary have attracted a lot of attention in the past few years due to their relevance in the formulation of several different physical models at a classical but, especially, at a quantum level. Among these, certainly noteworthy are those whose underlying ambient manifold is the $d$-dimensional anti-de Sitter spacetime $\mathbb{A}d\mathbb{S}_d$, which is a maximally symmetric solution of Einstein's equations with a negative cosmological constant. In the Poincar\'e patch, $\mathbb{A}d\mathbb{S}_d$ is conformal to a subset of Minkowski spacetime $(\bR^d,\eta)$, namely, $\mathring{\bH}^d\doteq\{(t,x_1,\dots,x_{d-2},z)\;|\;z>0\}$. This entails that $\mathbb{A}d\mathbb{S}_d$ has a conformal boundary, which is located at $z=0$ in $\bR^d$.  

The presence of a timelike boundary is the source of several complications especially in the study of quantum field theoretical models, since our current understanding of these theories is mainly based on the assumption that the underlying background is globally hyperbolic; see, {\it e.g.}, \cite[Def. 1.3.8]{Baer_2007}. Here, we consider the algebraic formulation of quantum field theory, which is a versatile, computationally efficient framework, which allows to discuss with mathematical rigor both free and perturbatively interacting theories; see, {\it e.g.}, \cite{brunetti2, Rejzner} for recent introductions. Focusing on the example of a real scalar field, which is the case considered in this work, we can summarize the algebraic approach as a two-step procedure. In the first step, one individuates a unital $*$-algebra of observables which encompasses all the structural properties of the underlying theory, such as locality, the dynamics and the canonical commutation or anti-commutation relations. This first element of the algebraic quantization scheme relies on a main ingredient, namely the difference between the advanced and retarded propagators of the operator which encodes the dynamics. On a globally hyperbolic spacetime these propagators exist, they are unique and, most notably, they are causally supported, see \cite{Baer_2007}. 

The second step consists of the identification of a quantum state, namely a normalized and positive, linear functional on the algebra, out of which one can recover the probabilistic interpretation proper of quantum theories via the Gelfand-Naimark-Segal theorem. Among the plethora of possible quantum states, one usually restricts the attention to the {\it Gaussian} or {\it quasi-free} ones, see \cite{brunetti2}, which are completely determined by the underlying two-point correlation function. This is a bi-distribution which codifies structural properties such as positivity, the equations of motion and the canonical commutation relations, the latter condition necessary if we restrict the attention to scalar fields. Nonetheless, in order to guarantee the existence of a covariant renormalization scheme as well as that the quantum fluctuations of all observables are finite, one needs to further restrict the class of admissible two-point correlation functions by imposing the {\em Hadamard condition}, see e.g. \cite{Fewster:2013lqa}. This can be stated as a constraint on the singular structure of the underlying two-point correlation function, say $\omega_2$, which is best codified using the language of microlocal analysis in terms of the wavefront set of $\omega_2$. In addition, as shown by Radzikowski in his seminal works \cite{Radzikowski_1996, Radzikowski_1996_1} this is fully equivalent to establishing that, on each geodesically convex neighbourhood of the underlying manifold, the integral kernel of $\omega_2$ has to acquire a specific expression, known as {\em local Hadamard form} which is unique up to smooth remainders as first characterized by Kay and Wald \cite{Kay_Wald:1991} (see also \cite{Moretti:2021}).

The situation is drastically different as soon as one drops the assumption of the underlying manifold being globally hyperbolic in the sense of \cite[Def. 1.3.8]{Baer_2007} and a timelike boundary is allowed to exist. All the results listed above are no longer automatically valid and, to avoid a lengthy and possibly tedious discussion we content ourselves with listing the two main difficulties: 
\begin{enumerate}
	\item Advanced and retarded propagators are no longer unique. If one considers a normally hyperbolic operator, such as the Klein-Gordon one, initial data on a Cauchy surface do not suffice to identify a unique solution unless they are supplemented with suitable boundary conditions. These can be selected imposing that we work with a closed system, a feature which translates at a mathematical level in a suitable Green's identity holding true. Yet, even after choosing an admissible boundary condition, existence and uniqueness of the advanced and retarded propagators are not guaranteed {\it a priori} nor it is automatic that, even if existent, they abide by the standard causal properties, 
	\item The notion of a Hadamard two-point correlation function $\omega_2$ is strongly tied to the propagation of singularities theorem, which guarantees that the singular structure of $\omega_2$ is localized along light cones. Yet, the presence of a timelike boundary drastically changes the picture, since one needs to account for the reflection of light rays at the boundary. This classical phenomenon has far-reaching consequences; in particular it lies outside the scope of the results of \cite{Radzikowski_1996, Radzikowski_1996_1}, leading to the necessity of establishing a novel notion of Hadamard states tailored to these scenarios.
\end{enumerate}

Although in the past decade these questions have been investigated and several progresses have been made, a fully cohesive and satisfactory picture is still missing. From the PDE perspective of classifying admissible boundary conditions and of establishing the existence of advanced and retarded propagators, if the underlying manifold is static, one can resort to boundary triples and spectral techniques to give a complete answer; see \cite{Dappiaggi-Drago_2019}. Yet, this procedure falls short of the additional goal of characterizing the support properties of the propagators, which needs to be proven separately in a case-by-case setting. In some instances, this can be achieved with energy estimates, see, {\it e.g.}, Theorem 6 in \cite{Juarez-Aubry_Weder:2020}, though a suitable energy functional needs to be defined for each boundary condition and finding one can be highly non-trivial. At the level of two-point correlation functions, the literature is quite vast, especially if one restricts the attention to anti-de Sitter spacetimes. In order to summarize succinctly the current status with reference to the problem investigated in this paper, also in this case we list the main results and open points:
\begin{enumerate}
	\item From the viewpoint of microlocal analysis, using $b$-calculus techniques, it has been proven a generalization of the propagation of singularities theorem, first considering Robin boundary conditions, see \cite{Gannot_2022}, and then more general scenarios in \cite{Dappiaggi-Marta_2020}.
	\item A notion of Hadamard states in terms of the wavefront set of the underlying two-point correlation function has been formulated by different research groups and it accounts both for the propagation of singularities along light directions and for their reflection at the timelike boundary, see, {\it e.g.}, \cite{Dappiaggi:2017wvj, Dappiaggi-Marta_2021, Gannot_2022, Wrochna}.
\end{enumerate}  
Yet, in the literature, a counterpart of the results due to Radzikowski is still missing, the main difficulty lying in the characterization of the integral kernel of the two-point correlation function in a geodesic neighborhood that intersects the boundary. Herein, one needs to account at the same time both for the standard propagation along light geodesics and for their reflection at the boundary. In addition, finding a characterization of the admissible two-point correlation functions both at a microlocal level and at that of integral kernel is of paramount relevance if one wants to fully exploit the effectiveness of the algebraic, perturbative approach to interacting quantum field theories, see \cite{Rejzner}.

In this work we start a programme aimed at solving these open issues for real scalar fields. We begin by considering the $d$–dimensional half–Minkowski spacetime, $\bH^d\doteq\{(t,x_1,\dots,x_{d-2},z)\;|\;z\geq 0\}$. On the one hand $\bH^d$ is the simplest Lorentzian manifold with a timelike boundary and, for this class of spacetimes; it plays the same r\^ole of the $d$-dimensional Minkowski spacetime for globally hyperbolic backgrounds. On the other hand, since our ultimate goal is to solve the problems discussed above on a generic curved background with a timelike boundary, one should keep in mind that, locally, any of these manifolds is diffeomorphic to an open subset of $\bH^d$. Hence, proving a counterpart of Radzikowski theorems on $\bH^d$ is a natural and necessary step. 

On top of $\bH^d$ we consider, a real, massive Klein–Gordon field $\Phi$ whose dynamics is ruled by the operator $P := \Box_\eta + m^2$, where $\Box_\eta$ is the D'Alembert wave operator while $m^2\geq 0$. We endow it with Robin boundary conditions at $z=0$,
\begin{equation}\label{Eq: RobinBC}
	(\partial_z + \kappa)\Phi\big|_{z=0}=0,\qquad \kappa\ge 0.
\end{equation}
Robin boundary conditions are of special relevance since, on the one hand, they include as special cases, those of Neumann ($\kappa=0$) and of Dirichlet ($\kappa\to\infty$) type. For these cases the advanced and retarded propagators as well as the two-point correlation functions for the ground state can be constructed explicitly from their counterpart on the full Minkowski spacetime $(\bR^d,\eta)$ employing the method of images, see, {\it e.g.}, \cite{Dappiaggi-Nosari_2016}. On the other hand, Robin boundary conditions are of great physical relevance since they do not lead to total reflection at the boundary, but they allow for the existence of a smooth tail propagating at $z=0$, see, {\it e.g.}, \cite{Bondurant_2005}. We observe that we restrict the attention to $\kappa>0$ since, with our sign conventions, if $\kappa<0$, the two-point correlation function exhibits unstable modes which are exponentially suppressed in space though exponentially growing in time, see \cite{Dappiaggi-Drago_2019}.

Our aim is twofold: First of all we construct advanced and retarded propagators of $P$ on $\bH^d$ subordinated to the boundary condition as per Equation \eqref{Eq: RobinBC}, proving uniqueness and, above all, causal support, hence extending the results of \cite{Dappiaggi-Drago_2019} with which we make contact in Section \ref{Sec: BV problem}. Secondly we discuss the notion of Hadamard states, both globally using microlocal techniques and locally working at the level of integral kernel. We prove that one can individuate on $\bH^d$ a two-point correlation function, denoted by $\omega_{2,\kappa}\in\mathcal{D}^\prime(\bH^d\times\bH^d)$ which is consistent with both notions. This allows us in Section \ref{Sec: Comparison of Hadamard states} to prove a counterpart of Radzikowski's theorems in the case under scrutiny.

Most notably, when working at the level of integral kernels, we can show that Hadamard states abiding by Robin boundary conditions as per Equation \eqref{Eq: RobinBC}, can be written as the sum of a smooth function and of a singular kernel, called {\em Robin-Hadamard parametrix} which is in turn built out of two constituents. The first is the same Hadamard parametrix which appears in the analysis of the Klein-Gordon equation on a globally hyperbolic spacetime, such as Minkowski, while, the second is a novel contribution which is defined out of $\sigma_{-}$, a \emph{reflected} counterpart of the standard Synge world function, associated to mirror points across $\partial \bH^d$. We can therefore summarize our main results as follows: 

\begin{itemize}
	\item[\ding{104}]\emph{Propagators for Robin boundary conditions.} Extending the application of the Robin-to-Dirichlet transform introduced by Bondurant and Fulling in \cite{Bondurant_2005} for the analysis of the two-dimensional scenario, we construct distributions $G_\kappa^\pm\in\mathcal{D}^\prime(\bH^d\times\bH^d)$ which are inverses of the Klein-Gordon operator $P$ and which satisfy the boundary condition as per Equation \eqref{Eq: RobinBC} in both entries. In particular, using spectral theory arguments, we prove uniqueness, whereas convolution techniques allow us to establish that $G^\pm_\kappa$ enjoy the support properties of advanced and retarded propagators. 
	\item[\ding{104}] \emph{Robin Hadamard states and local Robin parametrix.}
	Adapting the global Hadamard condition to $\bH^d$, see, {\it e.g.}, \cite{Dappiaggi:2017wvj}, we characterize admissible two–point functions $\omega_{2,\kappa}$ in terms of their wavefront set, here adapted to encompass the existence of a boundary using the language of $b$-calculus. In our framework we can establish the existence of a canonical choice for $\omega_{2,\kappa}$ which is obtained out of the Poincar\'e vacuum two-point correlation function. In addition we can show that the integral kernel of $\omega_{2,\kappa}$ has the \emph{local Robin Hadamard form}, that is its singular part is the sum of a standard Hadamard parametrix written in terms of the Synge world function $\sigma$ and of a reflected counterpart built out of $\sigma_{-}$. More precisely we establish that
		\begin{gather*} 
		\omega_{2,\kappa}(x,x') = \tilde{H}_\kappa(x,x^\prime)+W(x,x^\prime)=\notag\\
		= \lim\limits_{\epsilon\to 0^+}\frac{U(x,x^\prime)}{\sigma_\epsilon^{\frac{d-2}{2}}} + \delta_d V(x,x^\prime) \ln \left(\frac{\sigma_\epsilon}{\lambda^2}\right) + \frac{U'(x,x^\prime)}{\sigma_{-,\epsilon}^{\frac{d-2}{2}}} + \delta_d V'(x,x^\prime) \ln \left(\frac{\sigma_{-,\epsilon}}{\lambda^2}\right)+W(x,x'),
	\end{gather*}
	where $W$ is smooth, $\lambda$ is a reference length, $\delta_d=0$ if $d$ is odd and $1$ if $d$ is even, while $\sigma_{\epsilon}(x,x') = \sigma(x,x') + i \epsilon (t(x) - t(x')) + \epsilon^2$, $t$ being the time coordinate. In addition, the functions $U, V$ and $U^\prime, V^\prime$ can be written as a formal power series in $\sigma$ and $\sigma_-$ respectively, whose coefficients can determined in terms of a set of recursive transport equations with prescribed initial data in the first case and boundary conditions in the second. The first ones are nothing but the standard Hadamard recursion relations while the second ones account for the boundary conditions and for both reflection and propagation along $\partial\bH^d$. A similar idea has been recently investigated in \cite{Pitelli_2025} on spacetimes with conical singularities.
	\item[\ding{104}] \emph{Equivalence of global and local formulations; Feynman parametrices.}
	We extend the classic Radzikowski theorem on $\bH^d$ proving that the local Robin Hadamard form implies and it is implied by the global counterpart expressed in term of a wavefront set condition on the underlying two-point correlation function. Hence we  establish an equivalence between local and global Hadamard states in this boundary setting, which entails, as a byproduct, the possibility of defining a Feynman parametrix, the building block of the perturbative approach to interacting quantum field theories.	
\end{itemize}

\paragraph*{Synopsis --} In Section~2 we fix the geometric and analytical setting on half-Minkowski spacetime $\bH^d$, we introduce the reflection map $\iota_z$ and the associated reflected Synge world function $\sigma_{-}$. 

In Section~3 we review the defining notions associated to fundamental solutions, propagators and parametrices in the case with an empty boundary, the classical Hadamard expansion as well as the global/local equivalence on globally hyperbolic spacetimes, preparing the ground for the generalization to the case of a manifold with a timelike boundary.

In Section~4 we construct the advanced and retarded propagators for the Klein-Gordon operator $P$ with Robin boundary conditions. After defining the Robin-to-Dirichlet map $T_\kappa$, generalizing the work of Bondurant-Fulling \cite{Bondurant_2005}, as well as the kernel $\mathcal{L}_\kappa$ of its inverse, we prove via suitable convolution estimates the causal support properties and we establish uniqueness. 

In Section~5 we develop a notion of Hadamard two-point correlation function $\omega_{2, \kappa}$, tailored to Robin boundary conditions. We establish a global Hadamard condition on $\bH^d$ adapting the one already present in the literature and using the language of $b$-calculus. In particular we show that $\omega_{2,\kappa}$ can be written as the sum between the Neumann two-point function, $\kappa=0$, and a correction term which is less singular and it accounts for propagation along the boundary. Furthermore we prove that $\omega_{2, \kappa}$ has the desired wavefront set and we construct a \emph{local Robin Hadamard parametrix} whose singular part is a sum of standard and reflected Hadamard terms. These are determined as power series in both $\sigma$ and $\sigma_-$ whose coefficients abide by suitable recursive transport equations. We conclude by proving the equivalence between the local and global formulations of Hadamard states and by constructing a Feynman parametrix $G_\kappa^{\mathrm{F}}$.

\section{Geometric setting}\label{Sec: Geometric setting}

The goal of this section is to set the main conventions used in this work and to introduce succinctly the geometric structures that we shall be employing. More precisely we call \emph{$d-$dimensional Minkowski half-space} $(\mathbb{H}^d, \eta)$, $d\geq 2$, the subset of $\bR^d$, which, on standard Cartesian coordinates $(t, x_1, ..., x_{d-2}, z) \equiv (\underline{x},z) \in \mathbb{R}^d$, coincides with the upper half-space
    \begin{equation} \label{Eq: half Minkowski}
    \mathbb{H}^d := \{(\underline{x},z) \in \mathbb{R}^d : z\ge0 \}, \hspace{0.2cm} d \ge 2,
\end{equation}
endowed with a Lorentzian metric $\eta$ of signature $(+,\underbrace{-, \ldots, -}_{d-1})$. The associated line element reads
\begin{equation*}
    \label{Eq: line element H^d}
    ds^2_{\mathbb{H}^d} = dt^2 - \sum_{i=1}^{d-2} dx_i^2 - dz^2.
\end{equation*}
Note that the boundary $\partial \mathbb{H}^d :=\{(\underline{x},0) \in \mathbb{R}^d\}$ is isometric to Minkowski spacetime $(\mathbb{R}^{d-1}, \eta)$ in dimension $d-1$. As mentioned in the introduction this is the prototype of a globally hyperbolic spacetime with a timelike boundary in the sense of \cite{Ak_Hau_2020} and therefore it is the starting point to understand quantum field theories on this larger class of backgrounds.

\begin{remark}\label{Rem: Euclidean half space}
For later convenience, we denote by $(\mathbb{H}^{d-1}, \delta)$, $d\geq 2$, the subset of $\bR^{d-1}$, which, on standard Cartesian coordinates $(x_1, ..., x_{d-2}, z) \equiv (\mathsf{x},z) \in \mathbb{R}^{d-1}$, coincides with the upper half-space
\begin{equation} \label{Eq: half Euclidean}
	\mathbb{H}^{d-1} := \{(\mathsf{x},z) \in \mathbb{R}^{d-1} : z\ge0 \}, \hspace{0.2cm} d \ge 2,
\end{equation}
endowed with the Euclidean metric $\delta$ of signature $(-, \ldots, -)$. The associated line element reads
\begin{equation*}
	\label{Eq: line element H^d-1}
	ds^2_{\mathbb{H}^{d-1}} = - \sum_{i=1}^{d-2} dx_i^2 - dz^2.
\end{equation*}
This can be read as the model space for a constant time hypersurface in half-Minkowski spacetime via the embedding $\bH^{d-1}\mapsto\{t\}\times\bH^{d-1}\subset\bH^d$, for all $t\in\bR$.
\end{remark}

\begin{remark}\label{Rem: AdS}
It is worth mentioning that, in the mathematical physics literature, the idea of a background with a timelike boundary is often implicitly associated to the $d-$dimensional Anti-de Sitter spacetime $\mathbb{A}d\mathbb{S}_{d}$, $d >2$, which corresponds to maximally symmetric solutions of Einstein's field equations with a negative cosmological constant $\Lambda < 0$. Since this possesses closed, timelike curves which violate any strong causality assumption, it is customary to restrict the attention to the \emph{Poincaré patch} $P\mathbb{A}d\mathbb{S}_{d}$. This is isometric to $\mathbb{H}^d \setminus \partial \mathbb{H}^d$ and the associated line element in standard Cartesian coordinates $(t, x_1, ..., x_{d-2}, z)\in \mathbb{R}^{d-1} \times \mathbb{R}_{+}$ reads
\begin{equation}
\label{Eq: line element of PAdS}
ds^2_{P\mathbb{A}d\mathbb{S}_{d}} = \frac{l^2}{z^2}(dt^2 - \sum_{i=1}^{d-2} dx_i^2 - dz^2), \, \, l^2 = - \frac{(d-1)(d-2)}{2 \Lambda}.
\end{equation}
Being the metric in Equation \eqref{Eq: line element of PAdS} singular in the limit $z \rightarrow 0^+$, the Poincar\'e patch of Anti-de Sitter is \emph{not} a globally hyperbolic spacetime with a timelike boundary. Yet, if we consider a conformal transformation we recover $\bH^d\setminus\partial\bH^d$, see Equation \eqref{Eq: half Minkowski}. In other words $P\mathbb{A}d\mathbb{S}_d$ can be endowed with a conformal timelike boundary. Our analysis can be used to study models also on this class of backgrounds provided that one reformulate them on $\bH^d$ by means of a suitable conformal transformation.
\end{remark}

%

In this work we will be using some basic geometric functions, whose definitions are here recollected for the reader's convenience and to set the notation once and for all. For definiteness, we specialize the attention to the class of backgrounds we are interested in, although this is not strictly necessary in what follows.

\begin{definition}
\label{Def: Synge's world function}
Let $(\bH^d,\eta)$ be half-Minkowski spacetime as per Equation \eqref{Eq: half Minkowski}. Given $x \in\bH^d$ and any among its normal convex neighborhoods $\mathcal{O}_x \subseteq \bH^d$, see \cite[Ch. 6]{Lee_2018}, we call \emph{geodesic distance} between $x$ and $x' \in \mathcal{O}_x$ the length of the unique geodesic connecting the two points and we denote it by $s(x,x')$. Furthermore we define the {\bf Synge's world function} as
    \begin{equation}
        \label{Eq: Synge's world function}
        \sigma(x,x') := \frac{1}{2} s(x,x')^2= \frac{1}{2} \eta_{\mu \, \nu} (x-x')^{\mu} (x-x')^{\nu} = \frac{1}{2}  \big[(t-t')^2 - \sum_{j=1}^{d-2} (x_j-x'_j)^2 - (z-z')^2\big].
    \end{equation}
\end{definition}

We also stress that a very powerful analytic tool on half-Minkowski spacetime is the method of images. This relies on a counterpart of the Synge's world function which encodes the reflective nature of the boundary and which is here defined for later convenience using a geometric language. 

\begin{definition}
\label{Def: Reflected Synge's world function}
Consider $(\mathbb{R}^d, \eta)$ and define the reflection map $\iota_z: \mathbb{R}^d \rightarrow \mathbb{R}^d, (\underline{x},z) \mapsto (\underline{x}, -z)$. We call {\bf reflected Synge's world function}
\begin{equation}
\label{Eq: reflected Synge world function}
    \sigma_{-} (x,x') := ({\iota_{z'}}^{-1})^* \sigma (x,x') = (\iota_{-z'})^* \sigma(x,x'),
\end{equation}
whose expression in Cartesian coordinates is   
\begin{equation}
	\label{Eq: reflected Synge world function H^d}
	\sigma_- (x, x') = \frac{1}{2} \big[(t-t')^2 -\sum_{j=1}^{d-2} (x_j-x'_j)^2 - (z+z')^2  \big].
\end{equation}
\end{definition}

\noindent Observe that, by restriction, Equation \eqref{Eq: reflected Synge world function H^d} is meaningful also on $\bH^d$, while its geometric counterpart in Equation \eqref{Eq: reflected Synge world function} requires the full Minkowski spacetime. The Synge's world functions plays a prominent r\^ole in the theory of parametrices for normally hyperbolic operators as well as in that of Hadamard states for quantum fields on a curved background. Therefore we shall be using it extensively and we shall enjoy some of its notable properties. In the following remark we list those which are relevant to us, following \cite[Sec. 3.3]{Poisson_2011} to which we refer for a thorough analysis. We highlight that all identities are given both for $\sigma$ and for its reflected counterpart $\sigma_-$ and the derivation of the latter ones follows slavishly that of the former. Hence we omit it.

\begin{remark}
Consider a geodesically convex open neighborhood $\mathcal{O} \subset\bR^d$ and, bearing in mind Definition \ref{Def: Synge's world function} we denote by $\sigma_\mu := \partial_\mu \sigma$ and $(\sigma_-)_\mu := \partial_\mu \sigma_-$, where $\partial_{\mu}$ are the partial derivatives subordinated to a choice of local coordinates in $\mathcal{O}$. It holds that
	\begin{equation}
		\label{Eq: notable relations with sigma, sigma-}
		\begin{cases}
			\sigma^\mu \sigma_\mu = 2\sigma, \\
			\sigma_-^\mu (\sigma_-)_\mu = 2\sigma_-.
		\end{cases}
	\end{equation} 
	Furthermore, given any smooth function $f:\bR^d \times \bR^d \to\bR$ which depends only on $\sigma$ and $\sigma_-$, namely $f\equiv f(\sigma,\sigma_-)$, it holds that we can rewrite as follows the action of derivative operators:
	\begin{equation}
		\label{Eq: derivative operators with sigma, sigma- on Minkowski}
		\begin{cases}
			\partial_\mu = \sigma_\mu \frac{\partial}{\partial \sigma} + (\sigma_-)_\mu \frac{\partial}{\partial \sigma_-}, \\
			\Box_{\eta} = d \frac{\partial}{\partial \sigma} + 2\sigma \frac{\partial^2}{\partial^2 \sigma} + 2\sigma^\mu (\sigma_-)_\mu \frac{\partial^2}{\partial \sigma \partial \sigma_-} + d \frac{\partial}{\partial \sigma_-} + 2\sigma_- \frac{\partial^2}{\partial^2 \sigma_-}.
		\end{cases}
	\end{equation}
\end{remark}

\vskip .3cm

To conclude the section, we introduce a last convention: given $(\bH^d,\eta)$ and a geodesically convex open neighborhood $\mathcal{O}_x\subseteq\bH^d$ of a point $x \in\bH^d$, for any scalar, continuous function $B:\mathcal{O}_x\times\mathcal{O}_x\to\bC$, we can define its {\bf coinciding point limit}
\begin{equation}\label{Eq: Coinciding Point limits}
	[B]:\mathcal{O}_x\to\bC,\quad x^\prime\mapsto [B](x^\prime)\doteq B(x,x).
\end{equation}
In particular, as shown in \cite[Sec. 3]{Poisson_2011}, we can apply this definition to some notable functions constructed out of the geodesic distance:
\begin{equation}\label{Eq: Notable Coinciding Point Limits}
[\sigma_{\mu \, \nu}] = \eta_{\mu \, \nu}\quad\textrm{and}\quad [\sigma^{\mu}{}_{\mu}] = d,
\end{equation}
where $d$ is the dimension of the underlying manifold.

\section{Fundamental solutions and Parametrices}\label{Sec: Cauchy Initial Value Problem and Fundamental solutions}
In this section we recollect some well-established analytic results concerning wave-like operators and their associated solution theory, assuming in the process that the reader is already familiar with the basic notions of microlocal analysis. As main references for the background material at the heart of this part of our work we consider \cite{Baer_2007} and \cite{Hormander_1990} and, only for the sake of this section, we do not consider a single spacetime, but rather a whole class. More precisely by $(\mathcal{M}, g)$ we denote a globally hyperbolic spacetime, see \cite[Def. 1.3.8]{Baer_2007}, recalling that, by definition, it has an empty boundary, $\partial\mcM=\emptyset$ . On top of $\mcM$ we consider a real, massive, scalar field $\Phi:\mcM\to\bR$ and the associated Cauchy problem
\begin{equation}
	\label{Eq: Cauchy initial value problem}
	\begin{cases}
		P\Phi := \left(\Box_g + m^2 + \xi R\right)\Phi= f\\
		\Phi \vert_{\Sigma} = \Phi_0 \\
		\nabla_{\mathfrak{n}} \Phi \vert_{\Sigma} = \Phi_1,
	\end{cases}
\end{equation}
where $\Box_g := g^{\mu \, \nu} \nabla_{\mu} \nabla_{\nu}$ is the D'Alembert wave operator, $m^2 \ge 0$ is a parameter to be interpreted as squared mass of the field, while $\xi \in \mathbb{R}$ is a free coupling to the scalar curvature $R$ built out of $g$. In addition $f\in C^\infty_0(\mcM)\equiv\mathcal{D}(\mathcal{M})$ is source term, while $\Sigma \subset \mathcal{M}$ is a spacelike Cauchy hypersurface with $\mathfrak{n}$ its future-directed timelike unit normal vector field. For definiteness we choose the initial data $(\Phi_0, \Phi_1)$ both as elements lying in $\mathcal{D}(\mcM)$ although this specific assumption plays no r\^ole in our analysis.

\paragraph{Advanced and Retarded Propagators --} As a matter of fact, the solution theory for Equation \eqref{Eq: Cauchy initial value problem} is fully under control, see \cite{Baer_2007} and it is best understood using the theory of fundamental solutions. Most notably the following result holds true \cite[Thm. 3.3.1, Cor. 3.4.3 \& Prop 3.4.8]{Baer_2007}:

\begin{proposition}\label{Prop: Advanced and Retarded fundamental Solutions}
	Let $(\mcM,g)$ be a globally hyperbolic spacetime and let $P$ be as per Equation \eqref{Eq: Cauchy initial value problem}. Then there exists unique {\bf advanced $(-)$ and retarded (+) Green's operators} $\mathcal{G}^\pm:\mathcal{D}(\mcM)\to C^\infty(\mcM)$ such that these maps are sequentially continuous and
	\begin{itemize}
	\item[1.] $P\circ\mathcal{G}^\pm=\mathrm{id}|_{\mathcal{D}(\mcM)}$ and $\mathcal{G}^\pm\circ P|_{\mathcal{D}(\mcM)}=\mathrm{id}|_{\mathcal{D}(\mcM)},$
	\item[2.] for any $f\in\mathcal{D}(\mcM)$
	$$\mathrm{supp}(\mathcal{G}^\pm(f))\subseteq J^\mp(\mathrm{supp}(f)).$$
	\end{itemize}
\end{proposition}

\begin{remark}\label{Rem: Propagators}
	In view of Proposition \ref{Prop: Advanced and Retarded fundamental Solutions} and of the Schwartz kernel theorem, we can associate to $\mathcal{G}^\pm$ unique bi-distributions $G^\pm\in\mathcal{D}^\prime(\mcM\times \mcM)$. Henceforth we shall mainly work at this level and we shall still refer to them as {\em advanced and retarded propagators}. In addition we introduce the {\bf advanced-minus-retarded (causal) propagator} $G=G^--G^+$ which is such that, working at the level of integral kernels, 
	\begin{equation}\label{Eq: Smart Trick}
	G^-(x,x^\prime)=\Theta(t-t^\prime)G(x,x^\prime)\quad\textrm{and}\quad G^+(x,x^\prime)=-\Theta(t^\prime-t)G(x,x^\prime),
	\end{equation}
	where $\Theta$ denotes the Heaviside step function. In addition $G$ can be realized as the solutions of the following distributional initial value problem on $\mcM$:
	\begin{equation}\label{Eq: Initial Value Problem for G}
		\begin{cases}
			(P\otimes\mathbb{I})G=(\mathbb{I}\otimes P)G=0\\
			\left.G\right|_{t=t^\prime}=0\quad\textrm{and}\quad\left.\partial_tG\right|_{t=t^\prime}=\left.-\partial_{t^\prime}G\right|_{t=t^\prime}=\delta_\Sigma,
		\end{cases}
	\end{equation}
	where $\delta_\Sigma$ is the Dirac delta supported on the diagonal of $\Sigma\times\Sigma$, while the symbol $|_{t=t^\prime}$ is a concise notation to denote the pullback of $G\in\mathcal{D}^\prime(\mcM\times\mcM)$ to $(\{t\}\times\Sigma)\times(\{t\}\times\Sigma)$. Observe that such operation is well-defined on account of standard microlocal arguments, in particular \cite[Thm 8.2.4]{Hormander_1990}.
\end{remark}

\noindent Establishing explicit expressions for the advanced and retarded propagators is in most of the cases not possible, notable exception being static spacetimes where one can employ spectral techniques, see {\it e.g.} \cite{Dappiaggi-Drago_2019}, or maximally symmetric backgrounds. 

\begin{remark}
A concrete example of this last statement is the $d$-dimensional Minkowski spacetime $(\mathbb{R}^d, \eta)$, $d\geq 2$. In this case, using complex analytic techniques, the advanced and retarded propagators, denoted by $G^{\pm}_{\bR^d}$, can be computed also exploiting Remark \ref{Rem: Propagators}. In particular, working at the level of integral kernels, it turns out that the advanced-minus-retarded fundamental solution reads
\begin{equation*}
	G_{\bR^d}(x,x^\prime)=\int_{\mathbb{R}^{d-1}} d^{d-1}k \frac{e^{-i \boldsymbol{k} \cdot (\boldsymbol{x}-\boldsymbol{x}')}}{2 \sqrt{|\boldsymbol{k}|^2 + m^2}} \left( e^{i \sqrt{|\boldsymbol{k}|^2 + m^2}(t-t')} - e^{-i \sqrt{|\boldsymbol{k}|^2 + m^2}(t-t')}  \right),
\end{equation*}
where we adopt the short-hand notation $(t, \boldsymbol{x}) \equiv (t, x_1, ..., x_{d-1}) \in \mathbb{R}^d$. The above equation admits a closed form expression as
\begin{equation}
	G_{\bR^d}(x,x^\prime)=\begin{cases}
		(-1)^{k-1}\,\dfrac{\mathrm{sgn}(t-t^\prime)}{(2\pi)^{k}}\,
		\Bigg(\dfrac{\partial}{\partial\sigma}\Bigg)^{k-1}
		\Big[J_0\!\big(m\, \sqrt{2\sigma}\,\big)\,\Theta(\sigma)\Big], & d=2k, \\[1.2em]
		\,\dfrac{\mathrm{sgn}(t-t^\prime)}{(2\pi)^{\,k+\frac12}}\,
		\Big(\dfrac{m}{\sqrt{2\sigma}}\Big)^{k-\frac{1}{2}}
		J_{\,k-\frac{1}{2}}\!\big( m\, \sqrt{2\sigma}\,\big)\,\Theta(\sigma), & d=2k+1.
	\end{cases}	
	\label{Eq: Causal Propagator in d Minkowski}
\end{equation}
where $\sigma(x,x^\prime)$ is the Synge's world function which takes the form of Equation \eqref{Eq: Synge's world function}. In addition, $\mathrm{sgn}$ and $\Theta$ denote respectively the sign and the Heaviside function, while $J_\alpha$ is the Bessel function of first kind of order $\alpha$. The corresponding advanced and retarded propagators can be constructed using Equation \eqref{Eq: Smart Trick}. 
\end{remark}

\paragraph{Parametrices --} In place of looking for the explicit construction of propagators, a more efficient and informative tool to work with is that of a parametrix. Despite being a standard concept, we feel worth recalling succinctly its definition, see \cite{Garabedian_1964}.

\begin{definition}\label{Def: Parametrix}
		Let $(\mcM,g)$ be a globally hyperbolic spacetime and let $P$ be as per Equation \eqref{Eq: Cauchy initial value problem}. We call {\bf parametrix} associated to $P$ any $H\in\mathcal{D}^\prime(\mcM\times\mcM)$ such that 
		$$(P\otimes\mathbb{I})H=\delta+R_1\quad\textrm{and}\quad(\mathbb{I}\otimes P)H=\delta+R_2,$$
		where $R_1,R_2\in C^\infty(\mcM\times\mcM)$.
\end{definition}

Among the infinitely many parametrices associated with the Klein-Gordon operator $P$, four of them are know to be {\em distinguished} due to their specific singular structure codified in their wavefront set. Two of them $H^\pm$ are nothing but the propagators $G^\pm$ in Proposition \ref{Prop: Advanced and Retarded fundamental Solutions} up to a smooth remainder, while the other two, denoted by $H^F,H^{\bar{F}}$, are called the Feynmann and anti-Feynmann parametrices. The underlying theory has been developed in \cite{Duistermaat_1972} and it can be seen as one of the pillars of the modern formulation of quantum field theory on curved backgrounds. We shall not delve into the details, but we content ourselves with reporting the explicit structure of their wavefront set:
		\begin{subequations}
			\begin{equation}
			\text{WF}(H^\pm) = \{(x,k_x, x', k_{x'}) \in T^*(\mathcal{M} \times \mathcal{M})\setminus\{0\} \, | \, (x,k_x) \sim_{\pm} (x', -k_{x'}), k_x \ne 0 \} \cup \text{WF}(\delta_{2}), 
		\end{equation}
				\begin{equation}
			WF(H^{F/\overline{F}}) = \{(x,k_x, x', k_{x'}) \in T^*(\mathcal{M} \times \mathcal{M})\setminus\{0\} \, | \, (x,k_x) \sim_{F/\overline{F}} (x', -k_{x'}), k_x \ne 0 \} \cup \text{WF}(\delta_{2}).
		\end{equation}
		\end{subequations}
		In the first one $\sim_{\pm}$ entails that $x \in J^{\pm}(x')$ and that there exists a lightlike geodesic $\gamma$ connecting $x$ to $x'$ such that $-k_{x'}$ is the parallel transport of $k_x$ along $\gamma$, while $\delta_2$ is the Dirac delta supported on the diagonal of $\mathcal{M} \times \mathcal{M}$. In the second one, instead, $\sim_{F}$ ({\em resp.} $\sim_{\overline{F}}$) entails that there exists a lightlike geodesic $\gamma$ connecting $x$ to $x'$ such that $-k_{x'}$ is the parallel transport of $k_x$ along $\gamma$ and, if $x' \in J^{+}(x)$, then $k_x \triangleright 0$ (\emph{resp}. $k_x \triangleleft 0$), while if $x' \in J^{-}(x)$, then $k_x \triangleleft 0$ (\emph{resp}. $k_x \triangleright 0$). The symbol $\triangleright 0$ ({\em resp.} $\triangleleft 0$) signifies that a vector is causal and future ({\em resp.} past) directed.

\paragraph{Hadamard Two-Point Distributions --} The last key ingredient that we introduce does not stem from the theory of partial differential equations of hyperbolic type or from microlocal analysis. It finds its origin in the formulation of free quantum field theories on globally hyperbolic curved backgrounds. Its physical significance is vast and discussing these aspects would bring us far from the main goal of this work. Therefore, we refer an interested reader to \cite{brunetti2}, Chapter 5 in particular. Here we content ourselves with giving its mathematical definition, following mainly the seminal work of \cite{Radzikowski_1996,Radzikowski_1996_1}.

\begin{definition}\label{Def: Hadamard States}
	Let $(\mcM,g)$ be a globally hyperbolic spacetime and let $P$ be as per Equation \eqref{Eq: Cauchy initial value problem}. We say that $\omega_2\in\mathcal{D}^\prime(\mcM\times\mcM)$ is a {\bf Hadamard two-point function} if, denoting by $G$ the advanced-minus-retarded propagator as in Remark \ref{Rem: Propagators}
	\begin{enumerate}
		\item $(P\otimes\mathbb{I})\omega_2=(\mathbb{I}\otimes P)\omega_2=0$,
	\end{enumerate} 
and for all $f,f^\prime\in\mathcal{D}(\mcM)$
	\begin{enumerate}
	\setcounter{enumi}{1}
		\item $\omega_2(f,f^\prime)-\omega_2(f^\prime,f)=i G(f,f^\prime)$,
		\item $\Im\left(\omega_2(f,f^\prime)\right)=\frac{1}{2}\left(G(f,f^\prime)\right)$, where $\Im$ denotes the imaginary part,
		\item $|G(f,f^\prime)|^2\leq 4\omega_2(f,f)\omega_2(f^\prime,f^\prime)$.
	\end{enumerate} 
In addition the singular structure of $\omega_2$ is such that 
\begin{equation}
	\label{Eq: omega2 WFset}
	\text{WF}(\omega_2) = \{(x,k_x, x', -k_{x'}) \in T^*(\mathcal{M} \times \mathcal{M})\setminus\{0\}\, | \, (x,k_x) \sim (x', k_{x'}) \, \text{and} \, k_x \triangleright 0 \},
\end{equation}
where $\sim$ entails that there exists a lightlike geodesic $\gamma$ connecting $x$ to $x^\prime$ so that $k_{x^\prime}$ is the parallel transport along it of $k_x$. In addition, $k_x\triangleright 0$ entails that $k_x$ is causal and future-pointing.
\end{definition}

Bi-distributions $\omega_2$ abiding by Definition \ref{Def: Hadamard States} are also said to be of {\em global} Hadamard form and their existence is a well established result, see \cite[Chap. 5]{Brunetti_1996} and references therein. It is worth mentioning that the very same argument which leads to the proof of this statement highlights that, unless the background is highly symmetric, it is hard to give explicit examples of Hadamard states. As a matter of fact, especially in the theoretical physics literature, it is more customary to work with bi-distributions which are of \emph{local Hadamard form}, \textit{i.e.}, in every geodesically convex open neighbourhood $\mathcal{O} \subseteq \mathcal{M}$, their integral kernel reads
\begin{equation}
    \label{Eq: local Hadamard form}
    \omega_2(x,x') := \lim_{\epsilon \rightarrow 0^+} \frac{U(x,x')}{4 \pi \sigma_{\epsilon}^{\frac{d-2}{2}}(x,x')} + \delta_d V(x,x') \ln \frac{\sigma_{\epsilon}(x,x')}{\lambda^2} + W(x,x'), \, \, \lambda \in \mathbb{R}
\end{equation}
where $\lambda>0$ is a reference scale length while $\sigma_{\epsilon}(x,x') = \sigma(x,x') + i \epsilon (t(x) - t(x')) + \epsilon^2$. Here $\sigma$ is the Synge's world function, see \cite{Poisson_2011}, which is a generalization of Definition \ref{Def: Synge's world function}, while $t: \mathcal{M} \rightarrow \mathbb{R}$ is a temporal function, whereas
\begin{equation*}
    \delta_d = \begin{cases}
        1 \, \, \text{if $d$ is even}, \\
        0 \, \, \text{if $d$ is odd}.
    \end{cases}
\end{equation*}
The functions $U,V,W\in C^\infty(\mathcal{O}\times\mathcal{O})$ and, while $W$ encodes the local freedom in choosing a Hadamard states, the others are fully determined by the underlying geometry and equations of motion. More precisely, see \cite{Friedlander_1975}, starting from Equation 
\eqref{Eq: local Hadamard form}, it is convenient to expand $U$ and $V$ as  formal power series in $\sigma$, namely,
\begin{flalign}
	\label{Eq: U power series}
	U = \sum_{j=0}^{\infty} u_j(x, x') \left( \frac{\sigma}{\lambda} \right)^j\;\textrm{and}\; V = \sum_{j=0}^{\infty} v_j(x, x') \left( \frac{\sigma}{\lambda} \right)^j,\quad u_j,v_j \in C^{\infty}(\mathcal{O} \times \mathcal{O}).
\end{flalign}
By imposing the equations of motion, one obtains {\em the Hadamard recursion relations} for the coefficients $\{u_j\}_{j=0}^{\infty}$ and
$\{v_j\}_{j=0}^{\infty}$, see also \cite{Decanini:2005eg}. In \emph{even} dimensions, setting $u_{-1}=v_{-1}=0$, they read for $j \in \mathbb{N}_0\cup\{-1\}$,
\begin{equation}
\label{Eq: Hadamard recursion relations on M even}
      \begin{cases}
        Pu_j + (2j + 4 -d) \sigma^{\mu} \partial_{\mu} u_{j+1} + (j+1) \left( \sigma^{\mu}{}_{\mu} + 2j + 4 - 2d\right) u_{j+1} + \frac{(2-d)}{2} (\sigma^{\mu}{}_{\mu} -d) u_{j+1} = 0, \\
        [u_0]=1,\;\textrm{and}\;[u_{j+1}] = -\frac{ [Pu_j]}{(j+1)(2j+4-d)},  \\\\
        Pv_j + 2(j+1) \sigma^\mu \partial_\mu v_{j+1} + (j+1) (\sigma^{\mu}{}_{\mu} + 2j) v_{j+1} = 0, \\
        [v_0] = - \frac{[Pu_{\frac{d}{2} - 2}]}{(d-2)}\;\textrm{and}\;[v_{j+1}] = - \frac{[Pv_j]}{(j+1)(d+2j)},
    \end{cases}
\end{equation}
whilst in \emph{odd} dimensions they become 
\begin{equation}
    \label{Eq: Hadamard recursion relations on M odd}
     \begin{cases}
        Pu_j + (2j + 4 -d) \sigma^{\mu} \partial_{\mu} u_{j+1} + (j+1) \left( \sigma^{\mu}{}_{\mu} + 2j + 4 - 2d\right) u_{j+1} + \frac{(2-d)}{2} (\sigma^{\mu}{}_{\mu} -d) u_{j+1} = 0, \\
        [u_0]=1,\;\textrm{and}\;[u_{j+1}] = -\frac{ [Pu_j]}{(j+1)(2j+4-d)}.
    \end{cases}
\end{equation}
Here the bracket $[\cdot]$ is defined as per Equation \eqref{Eq: Coinciding Point limits}. Observe that Equations \eqref{Eq: Hadamard recursion relations on M even} and \eqref{Eq: Hadamard recursion relations on M odd} are a set of transport equations with prescribed initial conditions in terms of coinciding point limits. 

\begin{remark}\label{Rem: Asymptotic series}
	Although individually the Hadamard recursion relations admit a solution and in addition all the coefficients are symmetric functions \cite{Moretti_2000}, the series in Equation \eqref{Eq: U power series} are in general only asymptotic, unless the underlying background is analytic. In this case convergence to an analytic function can be established, see \cite[Thm 4.3.1]{Friedlander_1975}.
\end{remark}

\noindent Definition \ref{Def: Hadamard States} and Equation \eqref{Eq: local Hadamard form} are actually two sides of the same coin as established by the celebrated Radzikowski theorem \cite{Radzikowski_1996,Radzikowski_1996_1}.

\begin{theorem}\label{Thm: Radzikowski}
	Let $(\mcM,g)$ be a globally hyperbolic spacetime and let $P$ be as per Equation \eqref{Eq: Cauchy initial value problem}. Then a bi-distribution $\omega_2\in\mathcal{D}^\prime(\mcM\times\mcM)$ is of global Hadamard form if and only if it is of local Hadamard form. 
\end{theorem}

\begin{remark}
	One of the most important aspects of Equation \eqref{Eq: local Hadamard form} is the realization that, locally, the singular component of any Hadamard state as per Definition \ref{Def: Hadamard States} is fully determined by the underlying geometry and dynamics. For this reason it is customary to call {\bf Hadamard parametrix} the singular component in Equation \eqref{Eq: local Hadamard form}, \textit{i.e.},
	\begin{equation}
		\label{Eq: Hadamard parametrix local form}
		H_\lambda(x,x') := \lim_{\epsilon \rightarrow 0^+} \frac{U(x,x')}{4 \pi \sigma_{\epsilon}^{\frac{d-2}{2}}(x,x')} + \delta_d V(x,x') \ln \frac{\sigma_{\epsilon}(x,x')}{\lambda^2}, \, \, \lambda \in \mathbb{R}.
	\end{equation}
	In comparison to Definition \ref{Def: Parametrix}, this nomenclature might appear ambiguous since $H_\lambda$ does not identify a bi-distribution on the whole manifold $\mcM$, but rather only on any but fixed convex geodesic neighborhood $\mathcal{O}_x$. In addition $H_\lambda$ is a bi-solution of the equation of motion rather than an inverse up to a smooth contribution. Yet, since this nomenclature is ubiquitous in the literature, we shall employ it.
\end{remark}

Following this analysis we can combine point {\em 2.} in Definition \ref{Def: Hadamard States} with Equation \eqref{Eq: local Hadamard form} to infer a local expression for the integral kernel of $G$, the advanced-minus-retarded fundamental solution of the Klein-Gordon operator, which reads
\begin{equation}
	\label{Eq: local propagator form}
	G(x,x') := \begin{cases}
		\mathrm{sgn}(\Delta t)\frac{V(x,x')}{2}\Theta(\sigma(x,x^\prime)), & d=2, \\[1.2em]
		\mathrm{sgn}(\Delta t)\left(\frac{(-1)^{\,\frac{d-2}{2}}}{2(2\pi)^{\frac{d-2}{2}}}\,U(x,x')\delta^{\frac{d-4}{2}}(\sigma)+\pi V(x,x') \Theta(\sigma(x,x^\prime))\right), & d=2k,\;k>1, \\[1.2em]
		\mathrm{sgn}(\Delta t)\,\frac{(-1)^{\frac{d-1}{2}}}{2\pi^\frac{d}{2}}\frac{U(x,x^\prime)}{(2\sigma)^{\frac{d-2}{2}}(x,x^\prime)}\Theta(\sigma(x,x^\prime)),
        & d=2k+1.
	\end{cases}	 
\end{equation}

\paragraph{Feynman Parametrix --} Having established the concept of Hadamard two-point function as per Definition \ref{Def: Hadamard States}, we can now establish a different albeit equivalent characterization of the Feynman propagator for the Klein-Gordon operator on a globally hyperbolic spacetime $(\mcM,g)$, namely \cite{Rejzner}
\begin{equation}\label{Eq: Feynman propagator}
G^F :=  G^+-i\omega_2
\end{equation}
 where $\omega_2$ is as per Definition \ref{Def: Hadamard States} while $G^+$ is the retarded propagator as per Proposition \ref{Prop: Advanced and Retarded fundamental Solutions}. Observe that Equation \eqref{Eq: Feynman propagator} differs from part of the literature by a multiplicative constant $i$. Bearing in mind \cite{Duistermaat_1972}, there exists an associated parametrix $H^F \in \mathcal{D}'(\mathcal{M} \times \mathcal{M})$, the \emph{Feynman parametrix}, such that, on every but fixed geodesically convex open neighborhood $\mathcal{O} \subset \mathcal{M}$, the associated integral kernel reads
\begin{equation}
    \label{Eq: local form of the Feynman parametrix}
    H^F_\lambda(x,x') = \lim_{\epsilon \rightarrow 0^+} \frac{U(x,x')}{ \sigma_{\epsilon, F}^{\frac{d-2}{2}}(x,x')} + \delta_d V(x,x') \ln \frac{\sigma_{\epsilon, F}(x,x')}{\lambda^2},
\end{equation}
where $\sigma_{\epsilon, F}(x,x') := \sigma(x,x') + i \epsilon$. Here $\lambda \in \mathbb{R}$, $\sigma(x,x')$ is the Synge's world function on $(\mcM,g)$ while $U$ and $V$ abide by the Hadamard recursion relations as per Equations \eqref{Eq: Hadamard recursion relations on M even} and \eqref{Eq: Hadamard recursion relations on M odd}.

\section{Fundamental Solutions on half-Minkowski Spacetime}\label{Sec: BV problem}

In this section we formulate the problem that we investigate in this work, focusing henceforth our attention to half-Minkowski spacetime $(\bH^d,\eta)$ as in Equation \eqref{Eq: half Minkowski}. In particular we consider a real scalar field abiding by the Klein-Gordon equation with Robin boundary conditions and we adress the problem of constructing the advanced and retarded propagators proving in particular their causal support properties.

\subsection{Formulation of the Problem}\label{Sec: Formulation of the Problem}

On top of $(\bH^d,\eta)$ we consider a real, scalar field $\Phi:\bH^d\to\bR$ which abides by the Klein-Gordon equation
\begin{equation}
\label{Eq: KG equation}
    P\Phi=(\Box_\eta+m^2)\Phi=0,
\end{equation}
where $\Box_\eta$ is the D'Alembert wave operator, while $m^2\geq 0$. Yet, contrary to the scenario considered in Section \ref{Sec: Cauchy Initial Value Problem and Fundamental solutions}, $\partial\bH^d\neq\emptyset$ entails that the Cauchy problem is not well-defined and, in order to guarantee both existence and uniqueness of the solutions, one needs to supplement initial data with boundary conditions assigned at $\partial\bH^d$. 

This hurdle has therefore a deep impact on all those structures which we have introduced in Section \ref{Sec: Cauchy Initial Value Problem and Fundamental solutions} starting from an initial value problem, namely the advanced and retarded propagators $G^\pm$, {\it cf.} Equations \eqref{Eq: Smart Trick} and \eqref{Eq: Initial Value Problem for G}, and the Hadamard distributions, see Definition \ref{Def: Hadamard States} . Among the infinitely many boundary conditions which can be consistently assigned, in this work we shall consider those of {\bf Robin type} and, to start with, we give a definition of advanced and retarded propagators, tailored to this choice, although generalization to other boundary conditions are straightforward. Similarly the following concept can be extended almost verbatim to any globally hyperbolic spacetime with a timelike boundary, see \cite{Ak_Hau_2020, Dappiaggi-Drago_2019}.

\begin{definition}\label{Def: SolFond on boundary}
	Given half-Minkowski spacetime $(\bH^d,\eta)$ as per Equation \eqref{Eq: half Minkowski} as well as $P$, the Klein-Gordon operator as per Equation \eqref{Eq: Cauchy initial value problem}, we call {\bf advanced $(-)$ and retarded (+) propagators with Robin boundary conditions} any $G^\pm_\kappa:\mathcal{D}^\prime(\bH^d\times\bH^d)$ such that, denoting by $\mathring{\bH}^d=\bH^d\setminus\partial\bH^d$,
	\begin{equation}\label{Eq: Boundary and PDE}
	(P\otimes\mathbb{I})G^\pm_\kappa\vert_{\mathring{\bH}^d\times\mathring{\bH}^d}=\delta\vert_{\mathring{\bH}^d\times\mathring{\bH}^d}\quad\textrm{and}\quad(\mathbb{I}\otimes P)G^\pm_\kappa\vert_{\mathring{\bH}^d\times\mathring{\bH}^d}=\delta\vert_{\mathring{\bH}^d\times\mathring{\bH}^d},
	\end{equation}
	and
	\begin{equation}\label{Eq: Boundary conditions}
	(\partial_{\mathbf{n}}\otimes\mathbb{I})G^\pm_\kappa|_{\partial\bH^d}=(-\kappa\otimes\mathbb{I})G^\pm_\kappa|_{\partial\bH^d},\quad\kappa\in\bR,
	\end{equation}
	where $\partial_{\mathbf{n}}\equiv\partial_z$ denotes the derivative along the direction normal to $\partial\bH^d$ while $\vert_{\partial\bH^d}$ stands for the pull back of $\mathcal{G}^\pm_\kappa$ to $\partial\bH^d$ along the first entry. In addition we require that, for all $f\in\mathcal{D}(\mathring{\bH}^d)$,
	\begin{equation}\label{Eq: Support at the boundary}
	\textrm{supp}(G^\pm_\kappa(f))\subseteq J^\mp(\textrm{supp}(f)),
	\end{equation}
	where the partial evaluation is on the second entry.
\end{definition}

\begin{remark}\label{Rem: Dirichlet boundary conditions}
	Equation \eqref{Eq: Boundary conditions} can be generalized including formally $\kappa\to\pm\infty$ as the cases corresponding to Dirichlet boundary conditions, that is $G^\pm_{\infty}|_{\partial\bH^d}=0$.
\end{remark}

\noindent In the same spirit of Proposition \ref{Prop: Advanced and Retarded fundamental Solutions} it is convenient to introduce the {\em advanced-minus-retarded propagator} $G_\kappa\doteq G^-_\kappa-G^+_\kappa\in\mathcal{D}^\prime(\bH^d\times\bH^d)$ abiding by the following distributional mixed initial and boundary value problem: 

\begin{equation}
	\label{Eq: BV for G}
	\begin{cases}
		[P \otimes \mathbb{I}] G_\kappa= [\mathbb{I} \otimes P] G_\kappa = 0, \\
		G_\kappa \vert_{t=t'} = 0, \, \partial_t G_\kappa \vert_{t=t'} = - \partial_{t'} G_\kappa \vert_{t=t'} = \delta_{\mathring{\Sigma}},\\
		(\partial_{\mathbf{n}} G_\kappa + \kappa G_\kappa) \vert_{\partial\bH^d} = 0, \, \, \kappa \in \mathbb{R},
	\end{cases}
\end{equation}
where the symbol $\vert_{t=t'}$ denotes the pull back of $G$ on the initial value surface while $\mathring{\Sigma}=\Sigma\setminus\partial\Sigma$ denotes the interior of a constant time Cauchy surface of $\bH^d$. Observe that, since the first two lines in Equation \eqref{Eq: BV for G} correspond to Equation \eqref{Eq: Boundary and PDE}, we should be denoting a restriction of $G_\kappa$ to $\mathring{\bH}^d\times\mathring{\bH}^d$. Yet, we prefer avoiding it since it would make the notation too heavy and we feel that the meaning is clear from the context. As mentioned in the introduction, we reiterate two remarkable facts concerning advanced and retarded propagators on globally hyperbolic manifolds with a timelike boundary and abiding by Robin boundary conditions: 
\begin{itemize}
	\item[\ding{104}] existence of $G_\kappa$ has already been proven in \cite[Thm. 30]{Dappiaggi-Drago_2019} on every static, globally hyperbolic spacetime with a timelike boundary using spectral analytic techniques, but it is still unknown whether the ensuing propagators $G^\pm_\kappa$ obtained using Equation \eqref{Eq: Smart Trick} obey the causal support properties listed in Proposition \ref{Prop: Advanced and Retarded fundamental Solutions} unless $\kappa=0$ or $\kappa\to \pm \infty$. 
	\item[\ding{104}] there is no known counterpart for Equation \eqref{Eq: local propagator form}, in particular if one considers convex geodesic neighborhoods centered at a point $x \in\partial\bH^d$.
\end{itemize}

\noindent In the following we devise a procedure which will allows us to give an answer to both these questions. 

\begin{remark}\label{Rem: Positive kappa}
    In Equation \eqref{Eq: BV for G} we have formulated the Robin Boundary conditions for every value of $\kappa\in\bR$. As far we work at the level of fundamental solutions, this requirement is harmless, but, when considering instead two-point correlation functions as in Section \ref{Sec: Hadamard Recursion Relations}, for $\kappa<0$ modes exponentially growing in time are present and they lead to models which are considered to be unstable from the physical viewpoint. For this reason, henceforth we shall restrict our attention to $\kappa\geq 0$.
\end{remark}

\subsection{Fundamental solutions on \texorpdfstring{$\bH^d$}{Hd} with Robin boundary conditions}\label{Sec: Existence of Robin fundamental solutions} 

As explained in the introduction the strategy that we follow is based on several separate steps which we discuss in the following.

\paragraph{Dirichlet and Neumann boundary conditions --} Our first step consists of considering on $\bH^d$ Equation \eqref{Eq: BV for G} though with two distinguished boundary conditions: {\em Dirichlet} and {\em Neumann}, corresponding respectively to $\kappa\to\pm\infty$ and $\kappa=0$ in Equation \eqref{Eq: Boundary conditions}. Denoting for simplicity the associated propagators respectively by $G_D$ and $G_N$, we are imposing
\begin{equation}\label{Eq: D and N Boundary Conditions}
G_D\vert_{\partial\bH^d} = 0\quad\textrm{and}\quad\partial_z G_N\vert_{\partial\bH^d}, = 0
\end{equation}
where the pull-back to $\partial\bH^d$ has to be referred only to the first entry of the bi-distributions. Adapting to the case in hand the {\em method of images} the following statement holds true, see also \cite{Dappiaggi-Nosari_2016}:

\begin{proposition}\label{Prop: D and N Causal Propagator}
	Let $(\bH^d,\eta)$, $d\geq 2$, denote half-Minkowski spacetime as per Equation \eqref{Eq: half Minkowski}, $P$ the Klein-Gordon operator as per Equation \eqref{Eq: Cauchy initial value problem} and  $G_{\bR^d}$ the advanced-minus-retarded fundamental solution of $P$ on the whole Minkowski spacetime as per Equation \eqref{Eq: Causal Propagator in d Minkowski}. It holds that, working at the level of integral kernel, for every $(\underline{x}, z), (\underline{x'}, z')\in\bH^d$, the bi-distributions 
	\begin{flalign}
		\label{Eq: G_D}
		G_D (\underline{x}, z, \underline{x'}, z') &= G_{\bR^d} (\underline{x}- \underline{x'}, z - z') - G_{\bR^d} (\underline{x}- \underline{x'}, z + z'), \\ 
		\label{Eq: G_N}  
		G_N (\underline{x}, z, \underline{x'}, z') &= G_{\bR^d} (\underline{x}- \underline{x'}, z-z') + G_{\bR^d} (\underline{x} - \underline{x'}, z+z'),
	\end{flalign}
are the unique solution of Equation \eqref{Eq: BV for G} with boundary condition as in Equation \eqref{Eq: D and N Boundary Conditions}. 
\end{proposition}

\begin{proof}
	Observe that $G_{\bR^d} (\underline{x} - \underline{x'}, z+z')$ is the integral kernel of  $\left(\iota^*_z\otimes\mathrm{id}|_{\bR^d}\right)G_{\bR^d}$ where   
	$$\iota_z:\bR^d\to \bR^d\quad (\underline{x},z)\mapsto\iota_z(\underline{x},z)=(\underline{x},-z).$$
	Since $\iota_z$ is a discrete isometry of the Minkowski metric, it holds that 
	$$(P\otimes\mathbb{I})\circ\left(\iota^*_z\otimes\mathrm{id}|_{\bR^d}\right)G_{\bR^d}=\left(\iota^*_z\otimes\mathrm{id}|_{\bR^d}\right)\circ (P\otimes\mathbb{I})G_{\bR^d}=0,$$
	from which it descends that $G_{D/N}$ are distributional solutions of the Klein-Gordon equation in both entries. In addition, the boundary conditions as in Equation \eqref{Eq: D and N Boundary Conditions} are satisfied as one can infer per direct inspection. With reference to Equation \eqref{Eq: BV for G}, only the initial conditions are left to be checked. Yet, still Equation \eqref{Eq: Causal Propagator in d Minkowski} entails that, denoting by $\underline{x} \equiv (t,x_1,\dots,x_{d-2})$ and, similarly, by $\underline{x}^\prime \equiv (t', x_1', ..., x_{d-2}')$
	\begin{equation*}
			G_{\bR^d} (\underline{x}- \underline{x'}, z + z')|_{t=t^\prime}=0
	\end{equation*}
and
\begin{equation*}
\partial_tG_{\bR^d} (\underline{x}- \underline{x'}, z + z')|_{t=t^\prime}=-\partial_{t^\prime}G_{\bR^d} (\underline{x}- \underline{x'}, z + z')|_{t=t^\prime}=\delta(x_1-x^\prime_1)\dots\delta(x_{d-2}-x^\prime_{d-2})\delta(z+z^\prime)=0,
\end{equation*}
where we used that $z+z^\prime>0$. To conclude, uniqueness descends from the general statement in \cite[Thm. 30]{Dappiaggi-Drago_2019}.
\end{proof} 

\begin{remark}\label{Rem: Green's Operators for D and N}
	We observe that, dropping the assumption that $(\underline{x}, z), (\underline{x'}, z')\in\bH^d$, we can read $G_{D/N}$ as elements of $\mathcal{D}^\prime(\bR^d\times\bR^d)$ and, henceforth we shall denote by $\mathcal{G}_{D/N}$ the associated Green's operators. 
\end{remark}

\begin{remark}\label{Rem: D and N Propagators on Hd}
	Starting from Proposition \ref{Prop: D and N Causal Propagator} we can define, in the same spirit of Equation \eqref{Eq: Smart Trick},
	$$G^-_{D/N}=\Theta(t-t^\prime) G_{D/N}\quad\textrm{and}\quad G^+_{D/N}=-\Theta(t^\prime-t) G_{D/N}.$$
	These are advanced and retarded propagators as per Definition \ref{Def: SolFond on boundary}, since the support condition is inherited automatically from that of $G_{\bR^d}^{\pm}$, see Proposition \ref{Prop: Advanced and Retarded fundamental Solutions}. Observe in addition that, dropping the assumption that $(\underline{x}, z), (\underline{x'}, z')\in\bH^d$ in Proposition \ref{Prop: D and N Causal Propagator}, we can read $G^\pm_{D/N}\in\mathcal{D}^\prime(\bR^d\times\bR^d)$.
\end{remark}

\subsubsection{The Bondurant-Fulling map}\label{Bondurant-Fulling map}
Having discussed how to explicitly construct Dirichlet and Neumann fundamental solutions on $d-$dimensional half-Minkowski spacetime $(\mathbb{H}^d, \eta)$, we shall now tackle the problem of constructing a counterpart of $G_D$ and $G_N$ in Proposition \ref{Prop: D and N Causal Propagator} with generic Robin boundary condition. In this endeavor we shall make use of a notable map introduced in \cite{Bondurant_2005}. In addition to defining it, in the following we study its properties, especially in connection to its interplay with distributions, hence extending the analysis in \cite{Bondurant_2005}.

\vskip .2cm

\begin{definition}\label{Def: Bondurant-Fulling map}
	Given $(\mathbb{H}^d, \eta)$, $ d \ge 2$, as per Equation \eqref{Eq: half Minkowski}, denoting by
	\begin{equation}
		\label{Eq: Dirichlet smooth functions}
		C^{\infty}_D(\mathbb{H}^d) := \{ f \in C^{\infty}(\mathbb{H}^d) \, | \, f \vert_{z=0} := f(\underline{x}, 0) = 0 \},
	\end{equation}
	and by
	\begin{equation}
		\label{Eq: Robin smooth functions}
		C^{\infty}_{\kappa}(\mathbb{H}^d) := \{ f \in C^{\infty}(\mathbb{H}^d) \, | \, \partial_z f \vert_{z=0} = -\kappa f \vert_{z=0} \}, \, \, \kappa>0,
	\end{equation}
	we call {\bf Robin-to-Dirichlet (or Bondurant-Fulling) map}
	\begin{flalign}
		\notag 
		T_{\kappa} : C^{\infty}_{\kappa} (\mathbb{H}^d) &\longrightarrow C^{\infty}_D(\mathbb{H}^d) \\ f &\mapsto T_{\kappa}(f) := (\mathbb{I}_{\underline{x}} \otimes (\partial_z + \kappa \mathbb{I}_z)) f. \label{Eq: Operator T}
	\end{flalign}
\end{definition}

\noindent Note that the operator $T_{\kappa}$ is linear and its kernel is 
\begin{equation}
 \text{ker}(T_{\kappa}) = \{ f \in C^{\infty}_{\kappa}(\mathbb{H}^d) \, | \, f(\underline{x}, z) = c(\underline{x}) e^{-\kappa z}\}. 
\end{equation}
Hence, the linear map 
\begin{flalign}
\notag
    \tilde{T}_{\kappa}: D(\tilde{T}_{\kappa}) := \frac{C^{\infty}_{\kappa}(\mathbb{H}^d)}{\text{ker}(T_{\kappa})} &\longrightarrow \text{Ran}(T_{\kappa}) \subset C^{\infty}_D(\mathbb{H}^d) \\ \label{Eq: map tilde T kappa} [f] & \mapsto \tilde{T}_{\kappa} [f] := T_{\kappa} (f),
\end{flalign}
is bijective per construction and, as a consequence, it admits an inverse $\tilde{T}_{\kappa}^{-1}: \text{Ran}(\tilde{T}_{\kappa}) \rightarrow D(\tilde{T}_{\kappa})$. Since $C^{\infty}_D(\mathbb{H}^d) \subset \mathcal{D}'(\mathbb{H}^d)$, we can look for the corresponding distributional inverse. Denoting by $\mathcal{L}_\kappa := \tilde{T}_{\kappa}^{-1}$, we look for $\mathcal{L}_\kappa \in \mathcal{D}'(\mathbb{H}^d)$ such that 
\begin{equation}
    \label{Eq: fundamental solution of T kappa}
     \tilde{T}_{\kappa} \mathcal{L}_\kappa = \delta.
\end{equation}
This can be readily solved by
\begin{equation}
    \label{Eq: G_kappa}
    \mathcal{L}_\kappa (z) = \Theta(z) e^{-\kappa z},
\end{equation}
where, with a slight abuse of notation, we have omitted the dependence on the coordinates tangent to the boundary since they will play no r\^ole in the following discussion. Observe that, choosing $\kappa > 0$, $\mathcal{L}_\kappa \in \mathcal{S}'(\bR^d)$, where we have implicitly taken into account the natural inclusion $\bH^d\subset\bR^d$. Our strategy consists of using $\mathcal{L}_\kappa$ to construct tempered distribution on $\bR^d$ which automatically abide by Robin boundary conditions, hence codifying an inverse of the statement in Definition \ref{Def: Bondurant-Fulling map}. The first step consists of identifying distributions which obey Dirichlet boundary conditions.

\begin{definition}
\label{Def: Dirichlet distributions}
Given $\mathcal{D}^\prime(\bR^d)$, we establish the following subspaces:
\begin{itemize}
	\item[\ding{104}] the set of distributions which are constructed out of a reflection along the hyperplane $z=0$:
	\begin{equation}
		\label{Eq: D distributions}
		\mathcal{D}'_-(\mathbb{R}^d) := \{ u \in \mathcal{D}'(\mathbb{R}^d) \, | \, \exists v \in \mathcal{D}'(\mathbb{R}^d), u = v- \iota_z^* v \},\;\textrm{and}\;\mathcal{S}'_-(\mathbb{R}^d) := \mathcal{S}'(\mathbb{R}^d) \cap \mathcal{D}'_{-}(\mathbb{R}^d),
	\end{equation}
	where the map $\iota_z$ is as per Definition \ref{Def: Reflected Synge's world function},
	\item[\ding{104}] {\bf Dirichlet (tempered) distributions}, \emph{i.e.}, those admitting a vanishing pull-back at $z=0$:
	\begin{equation}
		\label{Eq: Dirichlet distribution}
\mathcal{D}'_0(\mathbb{R}^d) := \{ u \in \mathcal{D}'(\mathbb{R}^d) \, | \, j_0^* u = 0 \}\quad\textrm{and}\quad \mathcal{S}'_0(\mathbb{R}^d) := \mathcal{D}'_0(\mathbb{R}^d) \cap \mathcal{S}'(\mathbb{R}^d),
	\end{equation}
where
	\begin{equation}
		\label{Eq: immersion map}
		\notag
		j_{0}: \mathbb{R}^{d-1} \to \mathbb{R}^d,\quad 	\underline{x}\mapsto j_0(\underline{x}) := (\underline{x}, 0).
	\end{equation}
\end{itemize}  
\end{definition}

Having established a notion of tempered Dirichlet distributions, the next step consists of showing that one can consider their convolution with $\mcL_\kappa$ as per Equation \eqref{Eq: G_kappa}.

\begin{proposition}\label{Prop: Convolution}
    Let $u \in \mathcal{S}'_0(\mathbb{R}^d)$ as per Equation \eqref{Eq: D distributions}. Then, if $\kappa > 0$, the convolution $\mathcal{L}_\kappa \star u$ exists and, additionally, $\mathcal{L}_\kappa \star u \in \mathcal{S}'(\mathbb{R}^d)$.
\end{proposition}

\begin{proof}
	Observe that 
	   \begin{equation*}
		\widehat{\mathcal{L}_\kappa}(p_z) = \frac{1}{ip_z - \kappa},
	\end{equation*}
	which entails that, since, for all $u \in \mathcal{S}'_0(\mathbb{R}^d)$, $\hat{u} \in \mathcal{S}'(\mathbb{R}^d)$, then
\begin{equation*}
\widehat{\mathcal{L}_{\kappa}} \hat{u} \in \mathcal{S}'(\mathbb{R}^d),
\end{equation*}
since $\widehat{\mathcal{L}_\kappa}$ is smooth on the real axis and decaying at infinity. Hence, we can define $\mathcal{L}_\kappa \star u := \widecheck{\widehat{\mathcal{L}_{\kappa}} \hat{u}} \in \mathcal{S}'(\mathbb{R}^d)$. 
\end{proof}

\noindent Motivated by the above theorem, we can give the following definition. 

\begin{definition}
\label{Def: Robin tempered distributions}
We call $\mathcal{S}'_{\kappa}(\mathbb{R}^d)$ the space of {\bf Robin tempered distributions} defined as 
\begin{equation}
\label{Eq: Dirichlet tempered distributions}
    \mathcal{S}'_{\kappa}(\mathbb{R}^d) := \mathcal{L}_\kappa \star \mathcal{S}_0'(\mathbb{R}^d). 
\end{equation}
\end{definition}

\begin{lemma}
    Let $u \in \mathcal{S}'_{\kappa}(\mathbb{R}^d)$ as per Definition \ref{Def: Robin tempered distributions}. Then, $u$ satisfies Robin boundary conditions, that is, $j_0^* \tilde{T}_{\kappa} u = 0$. 
\end{lemma}

\begin{proof}
By definition, if $u \in \mathcal{S}'_{\kappa}(\mathbb{R}^d)$, there exists $v \in \mathcal{S}'_0(\mathbb{R}^d)$ such that $u = \mathcal{L}_\kappa \star v$. Since $\tilde{T}_{\kappa}$ is a differential operator, we can write
\begin{equation*}
    j_0^* \tilde{T}_{\kappa} u = j_0^*(\tilde{T}_{\kappa} (\mathcal{L}_\kappa \star v)) = j_0^* ((\tilde{T}_{\kappa} \mathcal{L}_\kappa) \star v) = j_0^* v = 0,
\end{equation*}
which completes the proof.
\end{proof}

\paragraph{Advanced and Retarded Propagators for Robin boundary conditions on $\bH^d$ --} In the remainder of the section, we shall apply the above results to identify the advanced and retarded propagators of the Klein-Gordon operator $P$ with Robin boundary conditions.  In the following it is more convenient to work at the level of Green's operators.  

\begin{theorem}\label{Thm: Constrution of Robin}
    Denote by $G^\pm_D$ the Dirichlet advanced $(-)$ and retarded $(+)$ propagators for the Klein-Gordon operator as per Remark \ref{Rem: D and N Propagators on Hd}. Then, 
    \begin{equation}
        \label{Eq: Robin fundamental solution via Fulling}
        G^\pm_\kappa := \left(\mathcal{L}_\kappa\otimes\delta\right) \star G^\pm_D \circ (\mathbb{I}\otimes \tilde{T}_{\kappa}),
    \end{equation}
    are propagators of $P$ abiding by Robin boundary conditions as per Equation \eqref{Eq: Boundary conditions}. 
\end{theorem}

\begin{proof}
Notice that, for any $f\in\mathcal{D}(\bR^d$), both $G^\pm_D(f)$ and $G^\pm_D\left(\tilde{T}_{\kappa}(f)\right)\in\mathcal{S}^\prime_0(\bR^d)$ where smearing against $f$ is a partial evaluation on the second entry, see Equation \eqref{Eq: Dirichlet tempered distributions}. Therefore Proposition \ref{Prop: Convolution} entails that the convolution in Equation \eqref{Eq: Robin fundamental solution via Fulling} is well-defined. Furthermore $G^\pm_\kappa(f)\in\mathcal{S}^\prime_\kappa(\bR^d)$ as per Equation \eqref{Eq: Dirichlet tempered distributions}. In order to conclude the proof 
\begin{equation*}
    (P\otimes\mathbb{I}) G^\pm_\kappa = \left(\mathcal{L}_\kappa\otimes\delta\right) \star (P\otimes\mathbb{I})G^\pm_D \circ (\mathbb{I}\otimes \tilde{T}_{\kappa}) =  \left(\mathcal{L}_\kappa\otimes\delta\right)\star\delta|_{\text{Diag}(\bR^d\times\bR^d)}\circ (\mathbb{I}\otimes \tilde{T}_{\kappa})=\delta|_{\text{Diag}(\bR^d\times\bR^d)},
\end{equation*}
where $\delta|_{\text{Diag}(\bR^d\times\bR^d)}$ denotes the delta supported on the diagonal of $\bR^d\times\bR^d$ and where we exploited that $[P,\tilde{T}_\kappa]=0$ entails that we can exchange the action of $P$ and $\mcL_\kappa$. In addition, we have also used that $PG^\pm_D=\delta|_{\text{Diag}(\bR^d\times\bR^d)}$. A similar reasoning entails that $(\mathbb{I}\otimes P) G^\pm_\kappa =\delta|_{\text{Diag}(\bR^d\times\bR^d)}$.
\end{proof}

\begin{remark}
	Observe that we have not stated at this moment that $G^\pm_\kappa$ in Equation \eqref{Eq: Robin fundamental solution via Fulling} are advanced or retarded fundamental solutions since we have no control on their support.
\end{remark}

Starting from Equation \eqref{Eq: Robin fundamental solution via Fulling}, we can derive a notable identity concerning $G^\pm_\kappa$, which extends the result of \cite{Bondurant_2005} in $d=2$. Henceforth, it will be convenient to work at the level of integral kernels since it will make the various steps more transparent. To set the notation we denote by $G_r\in\mathcal{D}^\prime(\bR^d\times\bR^d)$ the distribution with integral kernel
\begin{equation}
\label{Eq: GD in terms of G and G+}
   G_r(\underline{x} - \underline{x}', z+z')=\left(\iota^*_z\otimes\mathrm{id}|_{\bR^d}\right)G_{\bR^d}(\underline{x} - \underline{x}', z-z^\prime),
\end{equation}
where $G_{\bR^d}$ is the advanced-minus-retarded fundamental solution of the Klein-Gordon operator $P$ on $(\mathbb{R}^d, \eta)$ while $\iota_z$ is as per Definition \ref{Def: Reflected Synge's world function}. 

\begin{proposition}
\label{Prop: GR in terms of G+ and Gkappa}
   For all $\kappa\in [0,\infty)$, the integral kernel of $G_\kappa=G^-_\kappa-G^+_\kappa$ where $G^\pm_\kappa$ is as per Equation \eqref{Eq: Robin fundamental solution via Fulling} can be decomposed as
    \begin{equation}\label{Eq: Kappa G}
        G_\kappa = G_N - 2 \kappa (\mathcal{L}_\kappa\otimes\delta)\star G_r, 
    \end{equation}
    where $G_N$ is as per Equation \eqref{Eq: G_N}.
\end{proposition}

\begin{proof} 
In order to avoid a heavy notation, we shall omit the dependence on the $\underline{x}$ and $\underline{x}'$ variables since they play no r\^ole in what comes next. Combining Equations \eqref{Eq: Robin fundamental solution via Fulling} and \eqref{Eq: GD in terms of G and G+}, we have for all $f \in \mathcal{D}(\mathbb{R}^d)$,
\begin{equation}
\label{Eq: 2 contributions}
    \left((\mathcal{L}_\kappa \otimes \delta) \star (G_D(\tilde{T}_{\kappa}f))\right)(z) =   \left((\mathcal{L}_\kappa \otimes \delta) \star (G_{\bR^d}(\tilde{T}_{\kappa}f))\right)(z) -  \left((\mathcal{L}_\kappa \otimes \delta) \star (G_r(\tilde{T}_{\kappa}f))\right)(z),
\end{equation}
where we used Equation \eqref{Eq: G_D} together with Equation \eqref{Eq: GD in terms of G and G+}. We evaluate the two different contributions separately. Focusing on the first one, we can write at the level of integral kernels
\begin{equation}
   \left((\mathcal{L}_\kappa \otimes \delta) \star (G_{\bR^d}(T_{\kappa}f))\right)(z) = \int_{\mathbb{R}^2}  \, \mathcal{L}_\kappa(z') G_{\bR^d}(z -z'-\tilde{z})\left(f^\prime(\tilde{z})+\kappa f(\tilde{z})\right)d\tilde{z} \, dz' ,
\end{equation}
where we used that $G_{\bR^d}(z^\prime,\tilde{z})=G_{\bR^d}(z^\prime-\tilde{z})$. Using the change of variables $\tilde{z}\mapsto x\doteq z-z^\prime-\tilde{z}$ and then $z^\prime\mapsto y\doteq z-z^\prime-x$ we end up with 
\begin{gather*}
\left((\mathcal{L}_\kappa \otimes \delta) \star (G_{\bR^d}(T_{\kappa}f))\right)(z) = \int_{\mathbb{R}^2}\mathcal{L}_\kappa(z-x-y)G_{\bR^d}(x)\left(f^\prime(y)+\kappa f(y)\right)dx \, dy=\\
\int_{\mathbb{R}^2}-\left(\frac{d}{dy}-\kappa\right)\mathcal{L}_\kappa(z-x-y)G_{\bR^d}(x)f(y)dx \, dy\overset{\scriptstyle y \mapsto -y}{=}\\
=\int_{\mathbb{R}^2}\left(\frac{d}{dy}+\kappa\right)\mathcal{L}_\kappa(z-x+y)G_{\bR^d}(x)f(-y)dx \, dy=\int\limits_{\bR} G_{\bR^d}(y+z)f(-y)dy,
\end{gather*}
which entails
\begin{equation}
\label{Eq: 1st contribution}
    \left((\mathcal{L}_\kappa \otimes \delta) \star (G_{\bR^d}(T_{\kappa}f))\right)(z)=G_{\bR^d}(f)(z)
\end{equation}
For what concerns the second contribution to Equation \eqref{Eq: 2 contributions}, we can use the changes of variables $\tilde{z}\mapsto x\doteq z-z^\prime+\tilde{z}$ and $z^\prime\mapsto y\doteq x+z^\prime-z$ to obtain
\begin{gather*}
	\left((\mathcal{L}_\kappa \otimes \delta) \star (G_r(T_{\kappa}f))\right)(z) = \int_{\mathbb{R}^2}\mathcal{L}_\kappa(y+x-z)G_r(x)\left(f^\prime(y)+\kappa f(y)\right)dx \, dy=\\
	\int_{\mathbb{R}^2}-\left(\frac{d}{dy}-\kappa\right)\mathcal{L}_\kappa(y+x-z)G_{\bR^d}(x)f(y)dx \, dy.
\end{gather*}
Adding and subtracting to this expression $ 2 \kappa \mathcal{L}_\kappa \star G_r(f)$, we end up with
\begin{equation}
\label{Eq: 2nd contribution}
    \left((\mathcal{L}_\kappa \otimes \delta) \star (G_r(T_{\kappa}f))\right)(z)=-G_r(f)(z)+2 \kappa ((\mathcal{L}_\kappa \otimes \delta) \star G_r(f))(z).
\end{equation}
Combining together Equations \eqref{Eq: 1st contribution} and \eqref{Eq: 2nd contribution}, we obtain the sought result. 
\end{proof}

\noindent It is worth observing that, in the case $\kappa=0$ Equation \eqref{Eq: Kappa G} coincides with the Neumann advanced-minus-retarded propagator as expected. In the following we use this result to prove that $G^\pm_\kappa$ as per Equation \eqref{Eq: Robin fundamental solution via Fulling} abide by the support property encoded in Equation \eqref{Eq: Support at the boundary}. To this end we need the following two ancillary results:

\begin{lemma}\label{Lem: First Figata}
	Given the operator $\mcL_\kappa$ as per Equation \eqref{Eq: G_kappa} and denoting by $\delta(2\sigma_-)\in\mathcal{D}^\prime(\bR^d\times\bR^d)$ the bi-distribution whose integral kernel is $\delta((t-t^\prime)^2-(x_1-x^\prime_1)^2-\ldots-(x_{d-2}-x'_{d-2})^2+(z+z^\prime)^2)$, it holds that
	$$\text{supp}\left(\left(\mathcal{L}_\kappa \otimes\delta\right)\star\delta(2\sigma_-)\right)\cap(\bH^d\times\bH^d)\subset \text{supp}(\Theta(2\sigma_-))\cap(\bH^d\times\bH^d),$$
	where $\Theta$ denotes the Heaviside step function. 
\end{lemma} 

\begin{proof}
	We work at the level of integral kernels and, not to burden the notation, we omit to indicate explicitly the dependency on the coordinates relative to the directions tangent to the boundary. In addition, with a slight abuse of notation, for a reason which will become clear in a moment, we set $\tau^2=(t-t^\prime)^2-(x_1-x^\prime_1)^2-\ldots-(x_{d-2}-x'_{d-2})^2$.  Equation \eqref{Eq: Kappa G} entails that we are interested in 
	\begin{flalign*}
		 ((\mathcal{L}_\kappa\otimes\delta) \star \delta(2\sigma_-))(z,z') = - \int_{\mathbb{R}} d\tilde{z} \, \Theta(\tilde{z} - z) e^{-\kappa(z - \tilde{z})} \delta (\tau^2 - (\tilde{z} +z')^2),
	\end{flalign*}    
	which entails that $\tau^2\geq 0$, else $\delta(\tau^2-(\tilde{z}+z^\prime)^2)=0$. Bearing this in mind and considering $\tau^2>0$,
	\begin{flalign}
	 ((\mathcal{L}_\kappa\otimes\delta) \star \delta(2\sigma_-))(z,z') &= - \frac{1}{2|\tau|} \int_{\mathbb{R}} d\tilde{z} \, \Theta(\tilde{z} - z) e^{-\kappa(z-\tilde{z})}[\delta(\tilde{z} +z'-|\tau|) + \delta(\tilde{z}+z'+|\tau|)]\notag \\ 
		& =-\frac{1}{2|\tau|} \left[ \Theta(|\tau| -z -z') e^{-\kappa(-|\tau| + z + z')} + \Theta(-|\tau| -z-z')e^{-\kappa(|\tau|+z+z')} \right].\label{Eq: Aux2}
	\end{flalign}
Since the restriction to half-Minkowski spacetime entails that $z+z^\prime\geq 0$, it holds that on $\mathbb{H}^d\times\mathbb{H}^d$, 
$$((\mathcal{L}_\kappa\otimes\delta) \star \delta(2\sigma_-)))(z,z')=-\frac{1}{2|\tau|} \Theta(|\tau| -z -z') e^{-\kappa(-|\tau| + z + z')}=-\frac{1}{2|\tau|} \Theta(2\sigma_-) e^{-\kappa(-|\tau| +z+z')},$$
where we exploited that, on $\bH^d\times\bH^d$, $\Theta(2\sigma_-)=\Theta((|\tau| -z -z')(|\tau| +z +z'))=\Theta(|\tau| -z -z')$.

	In order to complete the analysis we need to investigate what happens as $\tau=0$. This can be done with a limiting procedure starting from Equation \eqref{Eq: Aux2}. Assuming that $z+z^\prime>0$, that is at least one of the two points lies in $\mathring{\bH}^d$, we obtain that, for $|\tau|$ sufficiently small $\Theta(|\tau| -z -z')$=0 and therefore the contribution in $(\mathcal{L}_\kappa\otimes\delta) \star \delta(2\sigma_-)$ vanishes. On the contrary, if $z+z^\prime=0$ such contribution is non-vanishing and it yields a term proportional up to smooth terms to $\frac{\Theta(\tau^2)}{|\tau|}$. Also in this case the sought statements descend and hence the proof is complete. 
\end{proof}

\begin{lemma}\label{Lem: Second Figata}
	Given the operator $\mcL_\kappa$ as per Equation \eqref{Eq: G_kappa} and denoting by $\Theta(2\sigma_-)\in\mathcal{D}^\prime(\bR^d\times\bR^d)$ the bi-distribution whose integral kernel is $\Theta((t-t^\prime)^2-(x_1-x^\prime_1)^2-\ldots-(x_{d-2}-x'_{d-2})^2+(z+z^\prime)^2)$, it holds that
	$$\text{supp}\left(\left(\mathcal{L}_\kappa \otimes\delta\right)\star\Theta(2\sigma_-)\right)\cap(\bH^d\times\bH^d)\subset \text{supp}(\Theta(2\sigma_-))\cap(\bH^d\times\bH^d),$$
	where $\Theta$ denotes the Heaviside step function.
\end{lemma} 

\begin{proof}
Also in this case we work at the level of integral kernels  and, for simplicity of the notation, we omit to indicate explicitly the dependency on the coordinates relative to the directions tangent to the boundary. In addition, with a slight abuse of notation, we set $\tau^2=(t-t^\prime)^2-(x_1-x^\prime_1)^2-\ldots-(x_{d-2}-x'_{d-2})^2$.  Equation \eqref{Eq: Kappa G} entails that we are interested in 
\begin{flalign*}
	((\mathcal{L}_\kappa\otimes\delta) \star \Theta(2\sigma_-))(z,z') = - \int_{\mathbb{R}} d\tilde{z} \, \Theta(\tilde{z} - z) e^{-\kappa(z-\tilde{z})} \Theta (\tau^2 - (\tilde{z} +z')^2),
\end{flalign*}    
which entails that $\tau^2\geq 0$, else $\Theta(\tau^2-(\tilde{z}+z^\prime)^2)=0$. We observe that $\Theta(\tilde{z}-z)$ entails $\tilde{z}>z > 0$ while $\Theta(\tau^2- (\tilde{z} +z')^2)$ entails that $\tilde{z}\in(-|\tau|-z^\prime,|\tau|-z^\prime)$. These conditions are compatible provided that $z<|\tau|-z^\prime$ which entails $\sigma_->0$ if we restrict the attention to $\bH^d\times\bH^d$ where $z+z^\prime>0$. As a consequence 
$$\Theta(\tilde{z} - z)\Theta (\tau^2 - (\tilde{z} +z')^2)=\Theta(|\tau|-(z+z^\prime))\Theta(\tilde{z} - z)\Theta (\tau^2 - (\tilde{z} +z')^2),$$
but, on half Minkowski spacetime $\Theta(|\tau|-(z+z^\prime))=\Theta(2\sigma_-)$ which completes the proof.
\end{proof}

In the following we discuss a notable example where Equation \eqref{Eq: Kappa G} can be made more explicit. This serves also as an inspiration for the analysis in Section \ref{Sec: Hadamard Recursion Relations}.

\begin{example}
\label{Ex: 4D massless case BF}
Let us consider the $4-$dimensional half-Minkowski spacetime $(\mathbb{H}^4, \eta)$ and the massless Klein-Gordon operator $m=0$. Observe that the two-dimensional scenario has been already analyzed in \cite{Bondurant_2005} and therefore we avoid discussing it. Working still at the level of integral kernels and omitting for simplicity of the notation to indicate all directions tangent to the boundary, Equation \eqref{Eq: Kappa G} entails that we are interested in evaluating
\begin{flalign*}
    2 \kappa (\mathcal{L}_\kappa \star G_r)(z,z') &= 2 \kappa \int_{\mathbb{R}} d\tilde{z} \, \mathcal{L}_\kappa (z, \tilde{z}) G_r(\tilde{z}, z') \\ 
    & = - \frac{2 \kappa}{\pi}\mathrm{sgn}(t-t^\prime) \int_{\mathbb{R}} d\tilde{z} \, \Theta(\tilde{z} - z) e^{-\kappa(z- \tilde{z})} \delta (\tau^2 - (\tilde{z} +z')^2).
\end{flalign*}    
We can therefore apply Lemma \ref{Lem: First Figata} to infer that
\begin{gather}
    2 \kappa (\mathcal{L}_\kappa \star G_r)(z,z') =\notag\\
    =-\frac{2\kappa}{\pi|\tau|}\mathrm{sgn}(t-t^\prime) \left[ \Theta(|\tau| -z -z') e^{-\kappa(-|\tau| +z+z')} + \Theta(-|\tau| -z-z')e^{-\kappa(|\tau|+z+z')} \right],\label{Eq: Aux1}
\end{gather}
where we set $|\tau| := +\sqrt{(t-t')^2 - (x-x')^2 - (y-y')^2}$. Hence the Robin propagator restricted to $\bH^d\times\bH^d$ reads
\begin{equation}\label{Eq: Gkappa sigma}
	G_\kappa(\tau, z,z') =\frac{\mathrm{sgn}(t-t^\prime)}{2\pi}\left(\delta(\sigma) + \delta(\sigma_-) - \frac{2\kappa}{|\tau|} \Theta(\sigma_-) e^{-\kappa(-|\tau| +z+z')}\right). 
\end{equation}
where $\sigma$ is the Synge's world function as per Definition \ref{Def: Synge's world function}, while $\sigma_-$ is the reflected counterpart as per Equation \eqref{Eq: reflected Synge world function H^d}. Observe that the limiting case $\tau=0$ can be dealt with as in Lemma \ref{Lem: First Figata}. To conclude this example, we make the following two important observations:
\begin{itemize}
	\item[\ding{104}] Equation \eqref{Eq: Gkappa sigma} entails that $G^-_\kappa=\Theta(t-t^\prime)G_\kappa$ and $G^+_\kappa=-\Theta(t^\prime-t)G_{\kappa}$ abide by Equation \eqref{Eq: Support at the boundary} and therefore they are retarded $(+)$ and advanced $(-)$ propagators in the sense of Definition \ref{Def: SolFond on boundary}.
	\item[\ding{104}] We can rewrite $G^-_\kappa$ as a series in powers of $\sigma_-$, namely Equation \eqref{Eq: Gkappa sigma} can be rewritten as
	\begin{gather}
		G^-_\kappa (t,\underline{x},t', \underline{x}') = \frac{\Theta(t-t')}{2\pi} \left[\delta(\sigma) + \delta(\sigma_-) - \frac{2\kappa}{|\tau|} \sum_{j=0}^{\infty} \frac{(-2\kappa)^j}{j!} \frac{1}{(|\tau| + z + z')^j} (\sigma_{-})^j \Theta(\sigma_-) \right]=\notag \\
        =\frac{\Theta(t-t')}{2\pi} \left[\delta(\sigma) + \delta(\sigma_-) + \Theta(\sigma_-)\sum_{j=0}^{\infty} d_j(z+z^\prime,\kappa)\,\sigma_-^{\,j}\right],
        \label{Eq: 4D massless case BF}
	\end{gather} 
\end{itemize}
where
$$d_j(z+z^\prime,\kappa)=2\kappa\sum_{n=0}^{j}A_n(z+z^\prime)\,E_{\,j-n}(z+z^\prime,\kappa),
\qquad
A_n(z+z^\prime)=\frac{(-1)^n}{(z+z^\prime)^{2n+1}}\frac{\binom{2n}{n}}{2^{\,n}},$$
while
$$E_q(z+z^\prime,\kappa)=\frac{1}{q!}\,B_q\!\big(1!c_1,\,2!c_2,\,\dots,\,q!c_q\big),
\quad
c_m(z+z^\prime,\kappa)=\frac{\kappa\,2^{\,m}}{(z+z^\prime)^{2m-1}}\binom{\tfrac12}{m}\ \ (m\ge1),
\quad E_0=1,$$
\(B_q\) being the complete Bell polynomials.
\end{example}

\noindent We conclude the section proving one of the main results of this paper, namely that Equation \eqref{Eq: Robin fundamental solution via Fulling} identifies advanced and retarded propagators as per Definition \ref{Def: SolFond on boundary}. For technical reasons we separate the odd from the even case, starting from the latter.

\begin{proposition}\label{Prop: Propagators in even dimension}
	Let $\bH^{d}$, $d\geq 2$ be the $d$-dimensional half Minkowski spacetime as per Equation \eqref{Eq: half Minkowski} and consider the massive Klein-Gordon operator $P$ as per Equation \eqref{Eq: Cauchy initial value problem}. Then $G^\pm_\kappa$ as per Equation \eqref{Eq: Robin fundamental solution via Fulling} are the advanced and retarded propagators for $P$ with Robin boundary conditions in the sense of Definition \ref{Def: SolFond on boundary}.
\end{proposition}

\begin{proof}
Let us consider Equation \eqref{Eq: Robin fundamental solution via Fulling} and in particular the representation as in Equation \eqref{Eq: Kappa G} of the associated advanced-minus-retarded fundamental solution. Per construction this is compatible with the dynamics, see Equation \eqref{Eq: Boundary and PDE} and with the boundary conditions, see Equation \eqref{Eq: Boundary conditions}. We need to check the causal support of the $G_\kappa$. This distribution is comprised of two terms. The first one is $G_N$ whose associated $G^\pm_N$ are compatible with Equation \eqref{Eq: Support at the boundary} as one can infer per direct inspection from Equation \eqref{Eq: G_N}.
Therefore we need only to focus on $\textrm{supp}((\mathcal{L}_\kappa\otimes\delta)\star G_r)$. Using Equation \eqref{Eq: Causal Propagator in d Minkowski}, we can derive an exact expression for $G_r$, but, for our purposes, it is only relevant that, in even dimensions, it can be written as a linear combination of terms which are proportional either to $\delta^{(n)}(2\sigma_-)$, $n\geq 0$ or to $\Theta(2\sigma_-)F(\sigma_-)$ where $F$ is suitable linear combination of Bessel functions which is smooth in its arguments. In odd dimensions, instead, only terms of the second kind appear. Yet
\begin{itemize}
	\item[\ding{104}] using the properties of the convolution $\left(\mathcal{L}_\kappa \otimes\delta\right)\star\delta^{(n)}(2\sigma_-)=\frac{d^n}{d\sigma_-^n}\left(\mathcal{L}_\kappa \otimes\delta\right)\star\delta(2\sigma_-)$. Lemma \ref{Lem: First Figata} entails that $\text{supp}\left(\left(\mathcal{L}_\kappa\otimes\delta\right)\star\delta(2\sigma_-)\right)\cap(\bH^d\times\bH^d)\subseteq\text{supp}(\Theta(2\sigma_-))\cap(\bH^d\times\bH^d)$. Since derivatives shrink at most the support, these terms are causally supported, 
	\item[\ding{104}] if we consider $\Theta(2\sigma_-)F(\sigma_-)$, then we are interested in controlling the support properties of 
	$$\left(\mathcal{L}_\kappa \otimes\delta\right)\star\Theta(2\sigma_-)F(\sigma_-),$$
	where $F$ is such that the convolution is still well-defined. Yet, we can apply Lemma \ref{Lem: Second Figata} since we can include $F$ in the statement on account of the observation that the proof relies only on the support properties of the Heaviside distributions, which are left unchanged. This reasoning applies in particular to $G_r$ in odd spacetime dimensions.
\end{itemize}
Gathering these pieces of information we can infer the sought conclusion.
\end{proof}

\noindent As a last step, we prove uniqueness of $G^\pm_\kappa$ and, to this end, we need to make contact with the spectral representation of the advanced and retarded propagators employed in \cite{Dappiaggi-Drago_2019}.

\begin{proposition}\label{Prop: Uniqueness}
The bi-distributions $G^\pm_\kappa\in\mathcal{D}^\prime(\bH^d\times\bH^d)$ as per Equation \eqref{Eq: Robin fundamental solution via Fulling} are the unique advanced and retarded propagators for the massive Klein-Gordon operator $P$ with Robin boundary conditions in the sense of Definition \ref{Def: SolFond on boundary}.
\end{proposition}

\begin{proof}
Starting from Equation \eqref{Eq: Kappa G}, we consider $G_N=G_{\bR^d}+(\iota^*_z\otimes\mathbb{I})G_{\bR^d}$. Setting the notation $x=(t,x_\perp,z)\in\bR^d$ and similarly for $x^\prime\in\bR^d$, we recall that on Minkowski spacetime we can rewrite $G_{\bR^d}$ in Fourier modes as
$$G_{\bR^d}(x,x^\prime)=\int\limits_{\bR^{d-2}}\frac{dk_\perp}{(2\pi)^{d-2}}e^{ik_\perp(x_\perp-x^\prime_\perp)}\int\limits_\bR \frac{dk_z}{2\pi} e^{ik_z(z-z^\prime)}\frac{\sin\omega(t-t^\prime)}{\omega},$$
where we are working at the level of integral kernels and where $\omega^2=|k_\perp|^2+k^2_z+m^2$. Observe that $k_\perp\in\bR^{d-2}$ and that, consequently,
$$G_r(x,x^\prime)=(\iota^*_z\otimes\mathbb{I})G_{\bR^d}(x,x^\prime)=\int\limits_{\bR^{d-2}}\frac{dk_\perp}{(2\pi)^{d-2}}e^{ik_\perp(x_\perp-x^\prime_\perp)}\int\limits_\bR\frac{dk_z}{2\pi} e^{ik_z(z+z^\prime)}\frac{\sin\omega(t-t^\prime)}{\omega}.$$
This allows to evaluate the remaining term, namely, up to the multiplicative factor $-2\kappa$
\begin{gather*}
(\mathcal{L}_\kappa\otimes\delta)\star G_r(x,x^\prime)=\int\limits_0^\infty d\tilde{z} e^{-\kappa\tilde{z}}\int\limits_{\bR^{d-2}}\frac{dk_\perp}{(2\pi)^{d-2}}e^{ik_\perp(x_\perp-x^\prime_\perp)}\int\limits_\bR dk_z e^{ik_z(z+z^\prime+\tilde{z})}\frac{\sin\omega(t-t^\prime)}{\omega}\\ =
\int\limits_{\bR^{d-2}} \frac{dk_\perp}{(2\pi)^{d-2}}e^{ik_\perp(x_\perp-x^\prime_\perp)}\int\limits_\bR\frac{dk_z}{2\pi} \frac{e^{ik_z(z+z^\prime)}}{ik_z - \kappa}\frac{\sin\omega(t-t^\prime)}{\omega}.
\end{gather*}
Putting together all these identities we obtain
\begin{gather*}
G_\kappa(x,x^\prime)=\int\limits_{\bR^{d-2}}\frac{dk_\perp}{(2\pi)^{d-2}}e^{ik_\perp(x_\perp-x^\prime_\perp)}\int\limits_\bR \frac{dk_z}{2\pi} \left(e^{ik_z(z-z^\prime)}+e^{ik_z(z+z^\prime)}\left(1-\frac{2\kappa}{\kappa-ik_z}\right)\right)\frac{\sin\omega(t-t^\prime)}{\omega}\\
=\int\limits_{\bR^{d-2}}\frac{dk_\perp}{(2\pi)^{d-2}}e^{ik_\perp(x_\perp-x^\prime_\perp)}\int\limits_0^\infty dk_z \Psi_\kappa(z)\overline{\Psi_\kappa(z^\prime)} \frac{\sin\omega(t-t^\prime)}{\omega},
\end{gather*}
where $\Psi_\kappa(z)=\frac{1}{\sqrt{2\pi}}\left(e^{-ik_zz}-\frac{\kappa+ik_z}{\kappa-ik_z}e^{ik_z z}\right)$. One can recognize that, given two test-functions $f,f^\prime\in\mathcal{D}(\mathring{\mcM})$, this last formula entails
\begin{gather*}
	G_\kappa(f,f^\prime)=\int\limits_{\bR^2}dt \, dt^\prime\left(f,A_\kappa^{-\frac{1}{2}}\sin(A^{\frac{1}{2}}_\kappa(t-t^\prime))f^\prime\right),
\end{gather*}
where $A_\kappa$ is the self adjoint extension of $\Delta+m^2$ on $(\bH^{d-1},\delta)$ with Robin boundary conditions, see Remark \ref{Rem: Euclidean half space}. Here $\Delta$ is the Laplace operator, while $(,)$ denotes the standard $L^2$-pairing on $\bH^{d-1}$. We can apply \cite[Thm. 30]{Dappiaggi-Drago_2019} to establish the sought uniqueness.   
\end{proof}

\section{Robin Hadamard States on \texorpdfstring{$\bH^d$}{Hd}}\label{Sec: Hadamard Recursion Relations}
\label{Hadamard recursion relations for Robin fundamental solutions}

\noindent The construction in Section \ref{Sec: Existence of Robin fundamental solutions} has led to the identification of advanced and retarded propagators on $\bH^d$, see Proposition \ref{Prop: Propagators in even dimension}, but, in comparison with the content of Section \ref{Sec: Cauchy Initial Value Problem and Fundamental solutions}, there are at least three key questions left open. First of all we have not identified a suitable counterpart of the notion of Hadamard state, see Definition \ref{Def: Hadamard States}, but, even more importantly, one is left to wonder also whether it is possible to devise a local Hadamard form, similar to that in Equation \eqref{Eq: Hadamard parametrix local form}, tailored to the underlying Robin boundary condition. This is strongly connected to the last open question, namely if such a local form also exists for the advanced-minus-retarded propagator, similarly to Equation \eqref{Eq: local propagator form} on a globally hyperbolic spacetime and if these local forms can be determined up to smooth functions which can be determined as asymptotic power series whose coefficients abide by recursion relations similarly to Equations \eqref{Eq: Hadamard recursion relations on M even} and \eqref{Eq: Hadamard recursion relations on M odd}.

While, at the level of propagators, Equation \eqref{Eq: G_kappa} combined with Lemmas \ref{Lem: First Figata} and \ref{Lem: Second Figata} would allow us to give directly a positive answer, for the economy of this work, we feel worth discussing all these issues simultaneously.

\subsection{Hadamard states with Robin boundary conditions}
Devising a notion of Hadamard states on a manifold with a timelike boundary, such as half-Minkowski spacetime, which could replace Definition \ref{Def: Hadamard States} has been thoroughly studied in the past years by different research groups both for a wider class of boundary conditions and for a more general set of backgrounds. The reason lies in the great interest towards models constructed on asymptotically anti-de Sitter spacetimes which possess a conformal timelike boundary, see Remark \ref{Rem: AdS}. In this section, we shall use the same geometric language and tools employed in \cite{Dappiaggi-Marta_2021}, summarized in Appendix \ref{Sec: GBB} although a reader is encouraged to consult also \cite{Dappiaggi:2017wvj,Dappiaggi-Marta_2020,Dybalski:2018egv,Gannot_2022,Wrochna} for further details. At an heuristic level, we recall that Hadamard states capture the singular structure of a distinguished class of physically sensible bi-distributions in terms of a constraint on the underlying wavefront set and the following definition is an adaptation to the case in hand of \cite[Def. 5.1]{Dappiaggi-Marta_2021}.

\begin{definition}\label{Def: Hadamard 2-pt function}
	Let $(\bH^d,\eta)$ be half-Minkowski spacetime as per Equation \eqref{Eq: half Minkowski}. A bi-distribution $\omega_2\in\mathcal{D}^\prime(\bH^d\times\bH^d)$ is called of {\bf global Hadamard form} if its restriction to $\mathring{\bH}^d$ has the following wavefront set:
	\begin{equation}\label{Eq: Hadamard_Wavefront_Set}
		\text{WF}( \omega_2 ) = \left\{ (x,k,x^\prime,-k^\prime) \in T^*(\mathring{\bH}^d\times\mathring{\bH}^d) \setminus \{ 0 \}\; |\; (x,k) \sim (x^\prime,k^\prime)\;\textrm{and}\; k \triangleright 0  \right\},
	\end{equation}
	where $\sim$ entails that $(x,k)$ and $(x^\prime,k^\prime)$ are connected by a generalized broken bicharactersitic, see Definition \ref{Def: generalized broken bicharacteristics}, while $k\triangleright 0$ means that the co-vector $k$ at $x\in\mathring{\mathbb{H}}^d$ is future-pointing. Furthermore we call $\omega_{2,\kappa}\in\mathcal{D}^\prime(\bH^d\times\bH^d)$ a {\em Hadamard two-point function} associated to $P$ with Robin boundary conditions, if, in addition to Equation \eqref{Eq: Hadamard_Wavefront_Set}, it satisfies
	\begin{flalign}
		\label{Eq: Equation of motion Robin}(P\otimes\mathbb{I})\omega_{2,\kappa}=(\mathbb{I}\otimes P)\omega_{2,\kappa}=0,\\
		(\partial_{\mathbf{n}}\omega_{2,\kappa} + \kappa \omega_{2,\kappa}) \vert_{\partial\bH^d} = 0, \, \, \kappa\geq 0,\label{Eq: State boundary condition}
	\end{flalign}
	and, for all $f,f^\prime\in\mathcal{D}(\mathring{\bH}^d)$,
	\begin{equation}\label{Eq: constraints on 2-pt function}
		\omega_{2,\kappa}(\bar{f},f)\geq 0,\;\;\text{\bf[Positivity]}\quad\textrm{and}\quad\omega_{2,\kappa}(f,f^\prime)-\omega_{2,\kappa}(f^\prime,f)=iG_\kappa(f,f^\prime)\;\;\text{\bf[CCR]},
	\end{equation}
	where $P$ is the Klein-Gordon operator as in Equation \eqref{Eq: Cauchy initial value problem}, while $G_\kappa$ is the advanced-minus-retarded propagator as in Equation \eqref{Eq: Kappa G}. 
\end{definition}

\begin{remark}
	Observe that, in Definition \ref{Def: Hadamard 2-pt function} we are apparently discarding Dirichlet boundary conditions since they correspond to $\kappa\to \infty$. Actually, this can be included modifying accordingly Definition \ref{Def: Hadamard 2-pt function}. For the sake of conciseness we avoid going through all the details, but none of the results that we obtain fails in this case.
\end{remark}

\noindent It is also natural to find a generalization to the case in hand of Definition \ref{Def: Parametrix} and of Equation \eqref{Eq: Hadamard parametrix local form}.

\begin{definition}\label{Def: Robin-Hadamard paraemtrix}
	We call {\bf Robin-Hadamard parametrix} any $H_\kappa\in\mathcal{D}^\prime(\bH^d\times\bH^d)$ with $\kappa\geq 0$ arbitrary but fixed, such that, given any $\omega_{2,\kappa}\in\mathcal{D}^\prime(\bH^d\times\bH^d)$, it holds that
	\begin{enumerate}
		\item $(P\otimes\mathbb{I})H_\kappa=(\mathbb{I}\otimes P)H_\kappa\in C^\infty(\bH^d\times\bH^d)$,
		\item $\omega_{2,\kappa}-H_\kappa\in C^\infty(\bH^d\times\bH^d)$, see Equation \eqref{Eq: Hadamard_Wavefront_Set},
		\item for all $f,f^\prime\in\mathcal{D}(\mathring{\bH}^d)$, $H_\kappa(f,f^\prime)-H_\kappa(f^\prime,f)=iG_\kappa(f,f^\prime)$ where $G_\kappa$ is as per Equation \eqref{Eq: G_kappa}.
	\end{enumerate}
\end{definition}

\noindent We conclude this section observing that existence of Hadamard states associated to the Klein-Gordon operator has been proven in \cite{Dappiaggi-Marta_2021} for a more general class of backgrounds and of boundary conditions using a deformation argument. In the case of half-Minkowski spacetime, we can adapt the construction of Section \ref{Sec: Existence of Robin fundamental solutions} to identify a canonical choice for a Hadamard two-point function. The starting point is the observation, investigated also in \cite{Dappiaggi-Nosari_2016} that, if we consider $\omega_2\in\mathcal{D}^\prime(\bR^d\times\bR^d)$, the two-point function of the unique Poincar\'e invariant ground state associated to the Klein-Gordon operator on Minkowski spacetime, then, adapting Equation \eqref{Eq: G_D}, the restriction to $\bH^d\times\bH^d$ of
\begin{equation}\label{Eq: Dirichlet Hadamard two-point function}
	\omega_{2,D}\doteq\omega_2-\left(\iota^*_z\otimes\mathrm{id}|_{\bR^d}\right)\omega_2,
\end{equation}
is a Hadamard two-point function abiding by Dirichlet boundary conditions, where $\iota_z$ is as per Definition \ref{Def: Synge's world function}. 

\begin{remark}
	Observe that, in this analysis, we are discarding the case $d=2$ and $m=0$, since, due to infrared singularities, there does not exist a Poincar\'e invariant vacuum state for a massless real scalar field on $(\bR^2,\eta)$. The ensuing analysis could be repeated choosing another Hadamard state, but we will avoid going through all the details of this case.
\end{remark}

\begin{proposition}\label{Prop: Hadamard 2-point function with Robin boundary conditions}
	Let $(\bH^d,\eta)$ be half-Minkowski spacetime for $d\geq 2$ as per Equation \eqref{Eq: half Minkowski} and let $P$ be the Klein-Gordon operator as per Equation \eqref{Eq: Cauchy initial value problem} barring the case $d=2$ and $m=0$. Let $\omega_{2,D}$ be the Hadamard state with Dirichlet boundary conditions as per Equation \eqref{Eq: Dirichlet Hadamard two-point function}. Then the restriction to $\bH^d\times\bH^d$ of
	\begin{equation}
		\label{Eq: omega2 kappa} \omega_{2,\kappa}\doteq(\mcL_\kappa\otimes\delta)\star\omega_{2,D}\circ(\mathbb{I}\otimes \tilde{T}_\kappa),\quad\kappa\geq 0
	\end{equation}
	is a Hadamard two-point function with Robin boundary conditions in the sense of Definition \ref{Def: Hadamard 2-pt function}. Here $\tilde{T}_\kappa$ and $\mcL_\kappa$ are defined as in Equation \eqref{Eq: map tilde T kappa} and \eqref{Eq: G_kappa}, respectively.
\end{proposition}

\begin{proof}
Repeating step by step the proof of Theorem \ref{Thm: Constrution of Robin} replacing $G^\pm_D$ with $\omega_{2,D}$ we can infer that $\omega_{2,\kappa}$ is a bi-solution of the equation of motion and using Proposition \ref{Prop: GR in terms of G+ and Gkappa} we can also infer that $\omega_{2,\kappa}(f,f^\prime)-\omega_{2,\kappa}(f^\prime,f)=iG_\kappa(f,f^\prime)$ for all $f,f^\prime\in\mathcal{D}(\mathring{\bH}^d)$. In order to control its wavefront set, we observe that the same reasoning used in Proposition \ref{Prop: GR in terms of G+ and Gkappa} entails that
\begin{equation}
   \label{Eq: omega2 Robin} \omega_{2,\kappa}=\omega_{2,N}+2\kappa\left(\mcL_\kappa\otimes\delta\right)\star\omega_{2,r},
\end{equation}
where 
\begin{equation}
    \label{Eq: omega2 Neumann}
    \omega_{2,N}=\omega_2+\left(\iota^*_z\otimes\mathrm{id}|_{\bR^d}\right)\omega_2,
\end{equation}
$\iota_z$ being as per Definition \ref{Def: Synge's world function}. At the same time $\omega_{2,r}\doteq\left(\iota^*_z\otimes\mathrm{id}|_{\bR^d}\right)\omega_2$. Working on $\bR^d$, we observe that $\text{WF}(\omega_{2,\kappa})\subset \text{WF}(\omega_{2,N})\cup\text{WF}(\left(\mcL_\kappa\otimes\delta\right)\star\omega_{2,r})$. Since $\iota_z$ is a discrete isometry of Minkowski spacetime, we can infer that $\text{WF}(\omega_{2,N})=\text{WF}(\omega_{2,r}) \cup \text{WF}(\omega_2)$, which in turn is determined by Equation \eqref{Eq: omega2 WFset}. If we consider that, $\left(\mcL_\kappa\otimes\delta\right)\star\omega_{2,r}$, on the one hand we observe that the convolution is well-defined being $\omega_{2,r}$ a tempered distribution, while, on the other hand $\mcL_\kappa$ acts in Fourier space as the multiplication by $\frac{1}{ik_z-\kappa}$. This is a smooth function and in the $z-$variable it is a pseudodifferential operator of order $-1$. Hence, by standard microlocal analysis techniques, $\text{WF}(\left(\mcL_\kappa\otimes\delta\right)\star\omega_{2,r})\subset\text{WF}(\omega_{2,r})$. Since $\left(\mcL_\kappa\otimes\delta\right)\star\omega_{2,r})$ cannot compensate any of the singularities of $\omega_{2,N}$, restricting to $\bH^d\times\bH^d$ and comparing with Equation \eqref{Eq: Hadamard_Wavefront_Set} yields the sought result. We are left with discussing positivity of $\omega_{2,\kappa}$ and, to this end, we employ the same spectral techniques used in Proposition \ref{Prop: Uniqueness}. More precisely, recalling that, working at the level of integral kernel, the two-point function of the Poincar\'e vacuum reads
$$\omega_2(x,x^\prime)=\int\limits_{\bR^{d-2}}\frac{dk_\perp}{(2\pi)^{d-2}}e^{ik_\perp(x_\perp-x^\prime_\perp)}\int\limits_\bR \frac{dk_z}{2\pi} e^{ik_z(z-z^\prime)}\frac{e^{i\omega(t-t^\prime)}}{\omega},$$
where $x=(t,x_\perp,z)$ and similarly $x^\prime = (t', x'_{\perp}, z')$, while $\omega^2=|k_\perp|^2+k^2_z+m^2$. Following the same steps as in in Proposition \ref{Prop: Uniqueness}, we end up with
\begin{gather*}
	\omega_{2,\kappa}(x,x^\prime)=\int\limits_{\bR^{d-2}}\frac{dk_\perp}{(2\pi)^{d-2}}e^{ik_\perp(x_\perp-x^\prime_\perp)}\int\limits_\bR dk_z \left(e^{ik_z(z-z^\prime)}+e^{ik_z(z+z^\prime)}\left(1-\frac{2\kappa}{\kappa-ik_z}\right)\right)\frac{e^{i\omega(t-t^\prime)}}{\omega}=\\
	=\int\limits_{\bR^{d-2}}\frac{dk_\perp}{(2\pi)^{d-2}}e^{ik_\perp(x_\perp-x^\prime_\perp)}\int\limits_0^\infty dk_z \Psi_\kappa(z)\overline{\Psi_\kappa(z^\prime)}\frac{e^{i\omega(t-t^\prime)}}{\omega},
\end{gather*}
where $\Psi_\kappa(z)=\frac{1}{\sqrt{2\pi}}\left(e^{-ik_zz}-\frac{ik_z+\kappa}{\kappa-ik_z}e^{ik_z z}\right)$. From this expression one can infer by direct inspection that $\omega_{2, \kappa}(\bar{f},f)\geq 0$ for every $f\in\mathcal{D}(\bH^d)$. 
\end{proof}

\subsection{Local Hadamard form and Recursion Relations}

In this section we investigate how to extend the notion of local Hadamard state which is essentially encoded on a globally hyperbolic spacetime in Equation \eqref{Eq: Hadamard parametrix local form}. Following the same strategy as in the previous sections, we give a formal definition which is valid in full generality and, subsequently, we investigate it on half-Minkowski spacetime.

\begin{definition}\label{Def: Local Hadamard State}
Let $(\bH^d,\eta)$ be half-Minkowski spacetime and let $\mathcal{O} \subseteq \bH^d$ denote a geodesically convex open subset. We call $\omega_{2,\kappa}\in\mathcal{D}^\prime(\bH^d\times\bH^d)$ a {\bf two-point function of local Hadamard form} associated to the Klein-Gordon operator $P$ with Robin boundary conditions, if, in addition to Equations \eqref{Eq: Equation of motion Robin}, \eqref{Eq: State boundary condition} and \eqref{Eq: constraints on 2-pt function}
\begin{itemize}
	\item[\ding{104}] $\omega_{2,\kappa}|_{\mathcal{O}\times\mathcal{O}}(x,x^\prime)$, the integral kernel of $\omega_{2,\kappa}$, is as in Equation \eqref{Eq: Hadamard parametrix local form} if $\mathcal{O}\cap\partial\bH^d=\emptyset$,
	\item[\ding{104}] if $\mathcal{O}\cap\partial\bH^d\neq\emptyset$, there exists $U,U',V,V',W\in C^\infty(\mathcal{O}\times\mathcal{O})$ such that, assuming $z+z^\prime>0$, 
	\begin{gather} 
		\omega_{2,\kappa}(x,x') = \tilde{H}_\kappa(x,x^\prime)+W(x,x^\prime)=\notag\\
		= \lim\limits_{\epsilon\to 0^+}\frac{U(x,x^\prime)}{\sigma_\epsilon^{\frac{d-2}{2}}} + \delta_d V(x,x^\prime) \ln \left(\frac{\sigma_\epsilon}{\lambda^2}\right) + \frac{U'(x,x^\prime)}{\sigma_{-,\epsilon}^{\frac{d-2}{2}}} + \delta_d V'(x,x^\prime) \ln \left(\frac{\sigma_{-,\epsilon}}{\lambda^2}\right)+W(x,x'),\label{Eq: Ansatz on the Robin parametrix}
	\end{gather}
	where $\lambda\in \mathbb{R}\setminus \{0\}$, $\sigma_\epsilon$ and similarly $\sigma_{-\epsilon}$ are defined as per Equation \eqref{Eq: local Hadamard form}, and
	\begin{equation*}
		\delta_d =
		\begin{cases}
			1 \, \, \text{if} \, \, d \, \, \text{is even} \\
			0\, \, \text{if} \, \, d \, \, \text{is odd}
		\end{cases}.
	\end{equation*}
\end{itemize}
The bi-distribution is $\tilde{H}_\kappa\in\mathcal{D}^\prime(\mathcal{O}\times\mathcal{O})$ is called {\bf local Robin Hadamard parametrix}.
\end{definition}

\noindent To start with, we prove that the two-point correlation function constructed in Definition \ref{Def: Local Hadamard State} is also locally of Hadamard form.

\begin{theorem}\label{Thm: State is of local Hadamard form}
	The two-point correlation function $\omega_{2,\kappa}\in\mathcal{D}^\prime(\bH^d\times\bH^d)$ as in Equation \eqref{Eq: 
    omega2 kappa} is of local Hadamard form. 
\end{theorem}

\begin{proof}
	We observe that Equations \eqref{Eq: Equation of motion Robin}, \eqref{Eq: State boundary condition} and \eqref{Eq: constraints on 2-pt function} hold true per construction. Hence, it suffices to prove Equation \eqref{Eq: Ansatz on the Robin parametrix} with $\mathcal{O}=\mathbb{H}^d$, since, if this is the case, then we can infer the sought conclusion by restriction to any other admissible subset $\mathcal{O}$. To this end, we consider the representation of $\omega_{2,\kappa}$ as in Equation \eqref{Eq: omega2 kappa}. Using Equation \eqref{Eq: omega2 Neumann} we can infer that $\omega_{2,N}$ is already of the desired form since the Poincar\'e vacuum $\omega_2$ is of local Hadamard form on the whole Minkowski spacetime, see Equation \eqref{Eq: local Hadamard form}. We can focus our attention directly on $(\mcL_\kappa\otimes\delta)\star\omega_{2,r}$ where, working at the level of integral kernels
	$$\omega_{2,r}(x,x^\prime)= \lim_{\epsilon \rightarrow 0^+} \frac{U(\iota_z(x),x')}{4 \pi \sigma_{-,\epsilon}^{\frac{d-2}{2}}(x,x')} + \delta_d V(\iota_z(x),x') \ln \frac{\sigma_{-,\epsilon}(x,x')}{\lambda^2} + \tilde{W}(\iota_z(x),x'), \, \, \lambda \in \mathbb{R} \setminus \{0\}
	$$
	where $\delta_d=0$ if $d$ is odd and $\delta_d=1$ if $d$ is even. Since $(\mcL_\kappa\otimes\delta)\star W$ is smooth being such $W$ itself, we can focus our attention on the singular contributions due to the logarithm and to the inverse powers of $\sigma_-$. We analyze these terms separately starting from the former. To simplify the notation we shall avoid writing explicitly the $\epsilon$-regularization and we shall work at the level of integral kernels. In addition we employ the notation $x=(\underline{x},z)$ and we write for simplicity $2\sigma_-(x,x^\prime)=\tau^2-(z+z^\prime)^2$ where $\tau^2=(t-t^\prime)^2-\sum\limits_{i=1}^{d-2}(x_i-x^\prime_i)^2$.
	
	\vskip .2cm
	
	\noindent{\em The $\ln\sigma_-$ term:} We are interested in evaluating convolutions of the form
	\begin{equation}\label{Eq: Aux12}
	(\mcL_\kappa\otimes\delta)\star (V\ln\sigma_-)=\int\limits_0^\infty e^{-\kappa s}V(\underline{x},z+s,x^\prime)\ln|\tau^2-(z+z^\prime+s)^2|ds,
	\end{equation}
	where the factor $\frac{1}{2}$ from the Synge world function has been omitted since it gives rise to a smooth contribution. We divide the analysis of the integral in two parts depending on $\tau^2$.
	\begin{itemize}
		\item[\ding{104}] If $\tau^2\leq 0$ and if $z+z^\prime>0$, the integrand is smooth and therefore the integral contributes to $W$ in Equation \eqref{Eq: Ansatz on the Robin parametrix}. 
		\item[\ding{104}] If $\tau^2>0$, then the only singular contribution comes from the points for which $\tau^2=(z+z^\prime+s)^2$ or, equivalently $s=s_*=-(z+z^\prime)+|\tau|$. Observe that we can consider the coordinate change $s\mapsto\alpha(s)=\tau^2-(z+z^\prime+s)^2$. This is well defined since $\partial_s\alpha(s)=-2(z+z^\prime+s)$ which is non vanishing since $s\geq 0$ and $z+z^\prime>0$. In addition $\alpha\in(-\infty,2\sigma_-]$ and the integral under scrutiny can be therefore rewritten as
		$$\int\limits_{-\infty}^{2\sigma_-}F(\alpha)\ln|\alpha| d\alpha,$$
		where $F(\alpha)=e^{-\kappa s(\alpha)}V(\underline{x},z+s(\alpha),x^\prime)\frac{\partial s}{\partial\alpha}$. Observing that $F$ is smooth and rapidly decaying at $-\infty$, we can introduce the primitive 
		$$K(x)=\int\limits_{-\infty}^{x}F(\alpha)d\alpha,$$
    which is also rapidly decreasing as $x\to-\infty$. A direct computation yields that
    \begin{flalign*}
        \int\limits_{-\infty}^{2\sigma_-}F(\alpha)\ln|\alpha| d\alpha &= \int\limits_{-\infty}^{2\sigma_-}F(\alpha)(\ln|\alpha| - \ln|2\sigma_-| + \ln|2\sigma_-|) d\alpha \\ &= K(2\sigma_-)\ln|2\sigma_-|+\int\limits_{-\infty}^{2\sigma_-}F(\alpha) \ln \left|\frac{\alpha}{2\sigma_-}\right| d\alpha.
    \end{flalign*}
    The second term in this last identity is locally integrable and it contributes with a logarithmic singularity.
	\end{itemize}
	
	\vskip .2cm
	
	\noindent{\em The $\sigma_{-}^{\frac{2-d}{2}}$ term:} We are interested in evaluating convolutions which read formally at the level of integral kernel
\begin{equation}\label{Eq: Aux14}
		(\mcL_\kappa\otimes\delta)\star (U\sigma_{-}^{\frac{2-d}{2}})=2^{\frac{d-2}{2}}\int\limits_0^\infty e^{-\kappa s}U(\underline{x},z+s,x^\prime)(\tau^2-(z+z^\prime+s)^2)^{\frac{2-d}{2}}ds.
	\end{equation}
	This requires to split the analysis between $d$ even and odd.
	\begin{itemize}
		\item[\ding{104}] If $d$ is even, setting $n=\frac{d-2}{2}\in\mathbb{N}$, we consider 
		$$(\mcL_\kappa\otimes\delta)\star (U\sigma_{-}^{-n}),$$
	where $\sigma_-^{-n}$ should be interpreted as the distribution $\text{PV}(\sigma_-^{-n}) + c_n\delta^{(n-1)}(\sigma_-)$ where $c_n\in\mathbb{C}$ is an irrelevant numerical factor. As before we observe that the integral yields a smooth contribution if $\tau^2<0$ or if $\tau^2=0$ and $z+z^\prime>0$. Therefore, we consider $\tau^2>0$ and we introduce a cut-off $\chi(s)$ supported in an arbitrarily small neighbourhood around $s_*=|\tau|-z-z^\prime$, which we denote by $(s_*-\epsilon,s_*+\epsilon)$, $\epsilon>0$. Adding in Equation \eqref{Eq: Aux14} a factor $1=\chi+(1-\chi)$ yields
	$$	(\mcL_\kappa\otimes\delta)\star (U\sigma_{-}^{\frac{2-d}{2}})=	\chi(\mcL_\kappa\otimes\delta)\star (U\sigma_{-}^{\frac{2-d}{2}})+	(1-\chi)(\mcL_\kappa\otimes\delta)\star (U\sigma_{-}^{\frac{2-d}{2}}).$$
	We start focusing on the first term. As in the analysis of the $\log$-term, we can consider the change of variables $s\mapsto\alpha(s)\doteq\tau^2-(z+z^\prime+s)^2\in(-\infty,2\sigma_-]$ which is well-defined on the whole domain of $s$, namely $(0,\infty)$. Working at the level of integral kernels this amounts to considering up to multiplicative numerical factors
	\begin{equation}\label{Eq: Aux15}
	\int\limits_{\alpha_-}^{\alpha_+}\chi(s(\alpha)) e^{-\kappa s(\alpha)}U(\underline{x},z+s(\alpha),x^\prime)\frac{ds}{d\alpha}\alpha^{-n}d\alpha=\int\limits_{\alpha_-}^{\alpha_+}G_\chi(\alpha)\alpha^{-n},
	\end{equation}
	where $\alpha_\pm=\pm 2|\tau|\epsilon-\epsilon^2$, while $G_\chi(\alpha)=\chi(s(\alpha))e^{-\kappa s(\alpha)}U(\underline{x},z+s(\alpha),x^\prime)\frac{ds}{d\alpha}\equiv\chi(\alpha)\tilde{G}(\alpha)$. Without loss of generality we choose $\epsilon$ so that $2|\tau|\epsilon-\epsilon^2>0$. The contribution due to $\delta^{(n-1)}(\sigma_-)=\delta^{(n-1)}(\frac{\alpha}{2})$ is smooth since the above integral is a formal expression for 
	$$\left(\delta^{(n-1)}\left(\frac{\alpha}{2}\right),G_\chi(\alpha)\right),$$
	where $G$ is a smooth function compactly supported on $(\alpha_-,\alpha_+)$. We can focus on the principal value recalling the distributional identity valid for any $\varphi\in C^\infty_0(\bR)$:
$$\left(\text{PV}(\alpha^{-n}),\varphi\right)=\frac{(-1)^{n-1}}{(n-1)!}\left(\frac{d^n\ln|\alpha|}{d\alpha^n},\varphi\right)=-\frac{1}{(n-1)!}\left(\ln|\alpha|,\frac{d^n\varphi}{d\alpha^n}\right).$$
The role of $\varphi$ is played by $G(\alpha)$ and we end up with
	$$
	\left( \mathrm{PV}\,\frac{1}{\alpha^n},\,G_\chi(\alpha) \right)
	=-\frac{1}{(n-1)!}\left(\ln|\alpha|,\frac{d^nG_\chi}{d\alpha^n}\right),$$
	which gives a smooth contribution. We are left with the analysis of the second term containing $(1-\chi)$. This yields two integrals, the first being 
	\begin{equation}\label{Eq: Aux18}
		\int\limits_{-\infty}^{\alpha_-}(1-\chi(s(\alpha))) e^{-\kappa s(\alpha)}U(\underline{x},z+s(\alpha),x^\prime)\frac{ds}{d\alpha}\alpha^{-n}d\alpha,
	\end{equation}
which yields a smooth contribution since the integrand decays exponentially at $\alpha\to-\infty$. Observe that the dependence on $\alpha_-$ is fictitious since it is reabsorbed by the one due to Equation \eqref{Eq: Aux15}. the second integral instead reads
\begin{equation}\label{Eq: AuxFinal}
	\int\limits_{\alpha_+}^{2\sigma_-}(1-\chi(s(\alpha))) e^{-\kappa s(\alpha)}U(\underline{x},z+s(\alpha),x^\prime)\frac{ds}{d\alpha}\alpha^{-n}d\alpha=\int\limits_{\alpha_+}^{2\sigma_-} \frac{G_{1-\chi}(\alpha)}{\alpha^n}d\alpha.
\end{equation}
Integrating repeatedly by parts with respect to the factor $\alpha^{-n}$ entails that this integral can be written as a linear combination of the from $\sum\limits_{k=1}^{n-1}\frac{c_k(x,x^\prime)}{\sigma_-^k}+c_0(x,x^\prime)\ln|\sigma_-|$ where the coefficients $c_k$ are smooth functions in both $x$ and $x^\prime$ for all $k=0,\dots,n-1$. As in the previous case the contribution form $\alpha_+$ is reabsorbed by the one due to Equation \eqref{Eq: Aux15}.
\item[\ding{104}] If $d$ is odd  we are interested in evaluating convolutions which read formally at the level of integral kernel
\begin{equation}
	(\mcL_\kappa\otimes\delta)\star (U\sigma_{-}^{\frac{2-d}{2}})=2^{\frac{d-2}{2}}\int\limits_0^\infty e^{-\kappa s}U(\underline{x},z+s,x^\prime)(\tau^2-(z+z^\prime+s)^2)^{\frac{2-d}{2}}ds.
\end{equation}
This is structurally identical to Equation \eqref{Eq: Aux14} and it can be dealt with in a similar way, hence we discuss it succinctly. Introducing the same cut-off function and coordinate change as in the case of $d$ even we have to deal with two contributions, one due to $\chi$ and one due to $(1-\chi)$. Working at the level of distributions, the first one reads
$$\left(\alpha^{\frac{2-d}{2}},G_\chi(\alpha)\right)=\left(\alpha^{\frac{2-d}{2}},G_\chi(\alpha)\right)=\left(C_d\frac{d^N|\alpha|^{\frac{1}{2}}}{d\alpha^N},G_\chi(\alpha)\right)=(-1)^NC_d\left(|\alpha|^{\frac{1}{2}},\frac{d^NG(\alpha)}{d\alpha^N}\right),$$
where $N=\frac{d-1}{2}$ while $C_d=\frac{\Gamma(\frac{3}{2}-N)}{\Gamma(\frac{3}{2})}$, $\Gamma$ being the Euler Gamma function. The second contribution is structurally identical to Equations \eqref{Eq: Aux18} and \eqref{Eq: AuxFinal} and, the last one in particular, yields a term of the form $\sum\limits_{k=1}^{\frac{d-1}{2}}\frac{c_k(x,x^\prime)}{\sigma_-^{k-\frac{1}{2}}}$ where $c_k(x,x^\prime)$ is smooth for all $k$.
\end{itemize}
Putting together all these data, the sought conclusion descends.
\end{proof}

\noindent This theorem shows that, when working with $\omega_{2,\kappa}$ the leading singularity is the one due to the Neumann part $\omega_{2,N}$, while the boundary condition affects only the other terms. In the following we shall prove that the functions $U,U',V,V'$ can be determined from the Klein-Gordon operator in terms of a power series either in $\sigma$ or in $\sigma_-$ whose coefficients abide by suitable recursion relations, which reflect the behaviour highlighted by Theorem \ref{Thm: State is of local Hadamard form}. In Section \ref{Sec: Comparison of Hadamard states} we shall compare Definitions \ref{Def: Hadamard States} and \ref{Def: Local Hadamard State} investigating to the case under scrutiny the result in \cite{Radzikowski_1996} valid on globally hyperbolic spacetimes.

\begin{remark}
	It is worth highlighting that in Equation \eqref{Eq: Ansatz on the Robin parametrix} the choice of the symbols $U$ and $V$ is not accidental since they will ultimately turn out to be structurally the same as those in Equation \eqref{Eq: Hadamard parametrix local form}. For this reason we felt unnecessary introducing additional symbols. 
\end{remark}

\subsubsection{Hadamard Recursion Relations: Even Dimensional Case}
\label{Sec: Hadamard Recurstion Relations even case}
We shall begin our analysis by focusing on Equation \eqref{Eq: Hadamard parametrix local form} in even dimensions. Mimicking the standard construction of the Hadamard recursion relations, see {\it e.g.} \cite{Garabedian_1964, Decanini:2005eg}, we consider on any geodesically convex opens subset $\mathcal{O}\subseteq\bH^d$, the following expansions
\begin{equation} \label{Eq: expansion of U,V,U',V'}
    \begin{cases}
    U(x,x^\prime) = \sum_{j=0}^{\frac{d-4}{2}} u_j (\underline{x},z, \underline{x}',z') \sigma^j \\
        V(x,x^\prime) = \sum_{j=0}^\infty v_j (\underline{x},z,\underline{x}',z') \sigma^j \\
        U'(x,x^\prime) = \sum_{j=0}^{\frac{d-4}{2}} u'_j (\underline{x},z, \underline{x}',z') \sigma_-^j \\
        V'(x,x^\prime) = \sum_{j=0}^\infty v'_j (\underline{x},z,\underline{x}',z') \sigma_-^j
    \end{cases},
\end{equation}
where the sum of the expansions of $U$ and $U'$ runs over a finite number of indices, since higher terms would yield a smooth contribution in Equation \eqref{Eq: Ansatz on the Robin parametrix} which could be reabsorbed in the definition of $W$.

Note that the action of the Klein-Gordon operator $P$ on the expansions of $U$ and $U'$ is formally the same, provided we replace $\sigma$ with $\sigma_-$ and the coefficients $\{u_j\}_{j=0}^{\frac{d-4}{2}}$, with their primed counterparts $\{u'_j\}_{j=0}^{\frac{d-4}{2}}$. An analogous statement holds true for the contributions $V \ln(\sigma)$ and $V' \ln(\sigma_-)$. Hence, it suffices to compute
\begin{flalign} 
\label{Eq: PU}
       \notag P\big[U(\sigma) \sigma^{-\frac{d-2}{2}} \big] &= (2-d) \sigma^\mu \partial_\mu u_0 \hspace{0.1cm} \sigma^{-\frac{d}{2}} \\ &+ \sum_{j=0}^{\frac{d}{2}-2} \big[Pu_j + (2j+4-d)\sigma^\mu \partial_\mu u_{j+1} + (j+1) (2j+4-d) u_{j+1} \big] \sigma^{j-\frac{d}{2}+1},
\end{flalign}
and 
\begin{flalign}
      \notag   P\big[ V(\sigma) \ln(\sigma) \big] &=  (PV)\ln(\sigma) + 2\eta^{\mu\nu}\partial_\mu V\partial_\nu \ln(\sigma) + V \Box \ln(\sigma) \\ \notag
        &= \big[2\sigma^\mu \partial_\mu v_0 + (d-2)v_0 \big] \sigma^{-1} + \\
        &+ \sum_{j=0}^\infty \big[Pv_j + 2(j+1)\sigma^\mu \partial_\mu v_{j+1} + (j+1)(d+2j)v_{j+1} \big] \sigma^j \ln(\sigma),
\end{flalign}
the other two contributions having the same functional dependence, replacing $\sigma$ with $\sigma_-$. Combining all terms together, we get
\begin{flalign}
\label{Eq: PHR}
 \notag P\tilde{H}_\kappa(\sigma, \sigma_-) &= (2-d) \sigma^\mu \partial_\mu u_0 \hspace{0.1cm} \sigma^{-\frac{d}{2}} \\ \notag
         &+ \sum_{j=0}^{\frac{d}{2}-3} \big[Pu_j + (2j+4-d)\sigma^\mu \partial_\mu u_{j+1} + (j+1) (2j+4-d) u_{j+1} \big] \sigma^{j-\frac{d}{2}+1} \\ \notag
         &+ \big[Pu_{\frac{d}{2}-2}+ 2\sigma^\mu \partial_\mu v_0 + (d-2)v_0 \big] \sigma^{-1} \\ \notag
         &+ \sum_{j=0}^\infty \big[Pv_j + 2(j+1)\sigma^\mu \partial_\mu v_{j+1} + (j+1)(d+2j)v_{j+1} \big] \sigma^j \ln(\sigma) \\ \notag
         &+ (2-d) \sigma_-^\mu \partial_\mu u'_0 \hspace{0.1cm} \sigma_-^{-\frac{d}{2}} \\ \notag
         &+ \sum_{j=0}^{\frac{d}{2}-3} \big[Pu'_j + (2j+4-d)\sigma_-^\mu \partial_\mu u'_{j+1} + (j+1) (2j+4-d) u'_{j+1} \big] \sigma_-^{j-\frac{d}{2}+1} \\ \notag
         &+ \big[Pu'_{\frac{d}{2}-2}+2\sigma_-^\mu \partial_\mu v'_0 + (d-2)v'_0 \big] \sigma_-^{-1} \\
         &+ \sum_{j=0}^\infty \big[Pv'_j + 2(j+1)\sigma_-^\mu \partial_\mu v'_{j+1} + (j+1)(d+2j)v'_{j+1} \big] \sigma_-^j \ln(\sigma_-).
\end{flalign}
In order for $P\tilde{H}_\kappa$ to lie in $C^\infty (\mathcal{O} \times \mathcal{O})$, one must require that the coefficients of the singular terms in Equation \eqref{Eq: PHR} vanish. Focusing on those related to $U$ and $V$, we obtain transport equations which need to be supplemented with suitable initial conditions. These can be expressed in terms of the coinciding point limits of the underlying coefficients, namely $[u_j]$ and $[v_j]$, see Equation \eqref{Eq: Coinciding Point limits}. In turn, since the equation for each of the coefficients $u_j, v_j$ depends on the preceding term $u_{j-1},v_{j-1}$ hence forming a set of recursion relations, consistency entails that only $[u_0]$ is free to be chosen. We select its value so to be the same as in half-Minkowski spacetime. 

Since the antisymmetric part of \eqref{Eq: Ansatz on the Robin parametrix} must coincide up to a multiplicative constant with $G_\kappa$, the advanced-minus-retarded propagator for the Klein-Gordon equation with Robin boundary conditions, we can read $[u_0]$ on half-Minkowski spacetime from Equation \eqref{Eq: Kappa G}. More precisely since Lemmas \ref{Lem: First Figata} and \ref{Lem: Second Figata} guarantee us that the term $2 \kappa (\mathcal{L}_\kappa\otimes\delta)\star G_r$ is proportional to the Heaviside function $\Theta(\sigma_-)$, we can safely conclude that the leading singularity coincides with that of the advanced-minus-retarded fundamental solution of the Klein-Gordon operator on the whole Minkowski spacetime. In other words, also taking into account the conventions of 
\cite{Bondurant_2005}, we can set $[u_0] = 1$. This entails that the {\em recursion relations} for $U$ and $V$ are 

\begin{equation} \label{Eq: relations for U even dimension}
    \begin{cases}
        (2-d) \sigma^\mu \partial_\mu u_0 = 0, \\
        [u_0] = 1, \\\\
        Pu_j + (2j+4-d)\sigma^\mu \partial_\mu u_{j+1} + (j+1)(2j+4-d)u_{j+1} = 0, \\
        [u_{j+1}] = -\frac{[Pu_j]}{2j+4-d}, \hspace{0.2cm} 0\le j \le \frac{d}{2}-3, 
    \end{cases}
\end{equation}    
\begin{equation} \label{Eq: relations for V even dimension}  
 	\begin{cases}
        Pu_{\frac{d}{2}-2} + 2\sigma^\mu \partial_\mu v_0 + (d-2)v_0 = 0, \\
        [v_0] = -\frac{[Pu_{\frac{d}{2}-2}]}{d-2}, \\\\
        Pv_j + 2(j+1) \sigma^\mu \partial_\mu v_{j+1} + (j+1)(d+2j) v_{j+1} = 0, \\
        [v_{j+1}] = - \frac{[Pv_j]}{(j+1)(d+2j)}, \hspace{0.2cm} j \in \mathbb{N}_{0}.
    \end{cases}
\end{equation}

\begin{remark}
	We observe that Equations \eqref{Eq: relations for U even dimension} and \eqref{Eq: relations for V even dimension} are the same Hadamard recursion relations that one obtains on globally hyperbolic spacetimes, see for example \cite{Decanini:2005eg}. This is consistent with the fact that, convex geodesic neighborhoods $\mathcal{O} \subset\bH^d$ such that $\mathcal{O} \cap\partial\bH^d=\emptyset$ are globally hyperbolic and, thereon, Equation \eqref{Eq: Ansatz on the Robin parametrix} takes the same form as Equation \eqref{Eq: local Hadamard form} since all terms depending on $\sigma_-$ are smooth. Hence they can be reabsorbed in the definition of $W$.
\end{remark}

If we focus our attention on Equation \eqref{Eq: PHR} relatively to the coefficients coming from the expansions $U',V'$, one can derive transport equations similar to those in Equations \eqref{Eq: relations for U even dimension} and \eqref{Eq: relations for V even dimension}. Yet, one cannot supplement them with initial conditions in terms of coinciding point limits, rather one needs to assign the behavior of the coefficients at $\partial\bH^d$. To this end, we evaluate 

\begin{flalign} 
\label{Eq: derivative BC}
 \notag       \partial_z \tilde{H}_\kappa &(\sigma, \sigma_-) \big|_{z=0} \simeq \partial_z \bigg[ U(\sigma) \sigma^{-\frac{d-2}{2}} + V(\sigma)\ln(\sigma) + U'(\sigma_-) \sigma_-^{-\frac{d-2}{2}} + V'(\sigma_-)\ln(\sigma_-) \bigg] \bigg|_{z=0} = \\ \notag & = \frac{1}{2} \big(2-d \big) \partial_z \sigma  \big(u_0 - u'_0 \big) \sigma^{-\frac{d}{2}} \big|_{z=0} \\ \notag &+\sum_{j=0}^{\frac{d}{2}-3} \bigg[ \partial_z u_j + \partial_z u'_j + \frac{1}{2} \big(2j+4-d\big) \partial_z \sigma \big( u_{j+1} - u'_{j+1} \big)\bigg] \sigma^{j-\frac{d}{2}+1} \bigg|_{z=0} \\
        &+ \bigg[ \partial_z (u_{\frac{d}{2}-2}+ u'_{\frac{d}{2}-2}) + \partial_z \sigma \big(v_0-v'_0 \big)\bigg] \sigma^{-1} \bigg|_{z=0} + \sum_{j=0}^\infty \bigg[ \partial_z (v_j\sigma^j) + \partial_z(v'_j\sigma_-^j) \bigg] \ln(\sigma) \bigg|_{z=0},
\end{flalign}
where $\simeq$ denotes that we are omitting on the right hand side all smooth contributions and
\begin{flalign} \label{Eq: k term in BC}
\notag        \kappa\tilde{H}_\kappa(\sigma, \sigma_-) \big|_{z=0} &= \kappa \bigg[u_0 + u'_0 \bigg] \sigma^{-\frac{d}{2}+1} \bigg|_{z=0} + \kappa \sum_{j=1}^{\frac{d}{2}-3} \bigg[u_j+ u'_j \bigg] \sigma^{j-\frac{d}{2}+1} \bigg|_{z=0} + \\
        &+ \kappa \bigg[u_{\frac{d}{2}-2} + u'_{\frac{d}{2}-2} \bigg] \sigma^{-1} \bigg|_{z=0} + \kappa \sum_{j=0}^\infty \bigg[v_j+v'_j \bigg]\sigma^j \ln(\sigma) \bigg|_{z=0},
\end{flalign}
where we exploited the fact that $\sigma \vert_{z=0} = \sigma_- \vert_{z=0}$ and $\partial_z \sigma_- \vert_{z=0} = - \partial_z \sigma \vert_{z=0}$. 
To derive the boundary conditions satisfied by the coefficients $\{u_j'\}_{j=0}^{\frac{d-4}{2}}$ and $\{v_j'\}_{j=0}^{\infty}$, we compare the same orders in powers of $\sigma_-$

\noindent In view of these considerations the Hadamard recursion relations for the coefficients $U^\prime,V^\prime$:
\begin{equation} 
\label{Eq: recursion for U', V'}
    \begin{cases}
        (2-d) \sigma_-^\mu \partial_\mu u'_0 = 0, \\
        u'_0 |_{z=0} = u_0 |_{z=0}, \\\\
        Pu'_j + (2j+4-d) \sigma_-^\mu \partial_\mu u'_{j+1} + (j+1)(2j+4-d) u'_{j+1} = 0, \\
        (\partial_z+\kappa) (u_j+u'_j) |_{z=0} + \frac{1}{2}(2j+4-d)\partial_z \sigma (u_{j+1}-u'_{j+1})|_{z=0} = 0 ,\hspace{0.2cm} 0 \le j \le \frac{d}{2}-3, \\\\
        Pu'_{\frac{d}{2}-2} + 2\sigma_-^\mu \partial_\mu v'_0 + (d-2) v'_0 = 0, \\
        \partial_z (u_{\frac{d}{2}-2} + u'_{\frac{d}{2}-2}) |_{z=0} + \partial_z \sigma (v_0-v'_0) |_{z=0} = - \kappa (u_{\frac{d}{2}-2}+u'_{\frac{d}{2}-2} )|_{z=0}, \\\\
        Pv'_j + 2(j+1) \sigma_-^\mu \partial_\mu v'_{j+1} + (j+1)(d+2j)v'_{j+1} = 0, \\
        (\partial_z+\kappa)(v_j+v^\prime_j)|_{z=0}+(j+1)\partial_z\sigma(v_{j+1}-v'_{j+1})|_{z=0}=0\hspace{0.2cm} j \in \mathbb{N}_0.
    \end{cases}
\end{equation}

\begin{remark}
	Observe that Equation \eqref{Eq: recursion for U', V'} is guaranteeing that $u_0=u^\prime_0$, which is tantamount to saying that the leading singularity of Equation \eqref{Eq: Ansatz on the Robin parametrix} is the same as that of $\omega_{2,N}$ while boundary conditions contribute to lower order terms. In addition the last recursion relation entails that, at $z=0$ where $\partial_z\sigma=z^\prime$, for all $j\in\mathbb{N}_0$,
    \begin{equation}\label{Eq: initial condition for u}
    u^\prime_{j+1}|_{z=0}=u_{j+1}|_{z=0}+\frac{(\partial_z+\kappa)(u_j+u^\prime_j)|_{z=0}}{(j+2-\frac{d}{2})z^\prime},
    \end{equation}
    and
    \begin{equation}\label{Eq: initial condition for v}
    v^\prime_{j+1}|_{z=0}=v_{j+1}|_{z=0}+\frac{(\partial_z+\kappa)(v_j+v^\prime_j)|_{z=0}}{(j+1)z^\prime}.
    \end{equation}
\end{remark}

To conclude the section we remark that the recursion relations that we have derived can also be applied at the level of the advanced-minus-retarded propagator as one can infer directly from Equation \eqref{Eq: State boundary condition}. We summarize this statement in the following corollary.

\begin{corollary}\label{Cor: Fundamental Solutions}
	Let $G_\kappa\in\mathcal{D}^\prime(\bH^d\times\bH^d)$ be the advanced-minus-retarded propagator for the Klein-Gordon operator with Robin boundary conditions in even spacetime dimensions on half-Minkowski spacetime. Then 
	\begin{itemize}
		\item[\ding{104}] if $\mathcal{O}\cap\partial\bH^d=\emptyset$,  $G_{\kappa}|_{\mathcal{O}\times\mathcal{O}}(x,x^\prime)$, the integral kernel of $G_\kappa$, reads 
		$$G_\kappa(x,x^\prime)= \beta_d U(x,x^\prime)\delta^{\frac{d-4}{2}}(\sigma) + \beta^\prime_d V(x,x^\prime)\Theta(\sigma),$$
		where $\beta_d=\frac{(-1)^{\,\frac{d-2}{2}}}{2(2\pi)^{\frac{d-2}{2}}}$ if $d\geq 3$ whereas $\beta_d=0$ if $d=2$. At the same time $\beta^\prime_d=\pi$ if $d\geq 3$ while $\beta^\prime_d=\frac{1}{2}$ if $d=2$,
		\item[\ding{104}] if $\mathcal{O}\cap\partial\bH^d\neq\emptyset$, then
		\begin{gather} 
			G_\kappa(x,x') = \tilde{H}_\kappa(x,x^\prime)+W(x,x^\prime)=\notag\\
			= \beta_d U(x,x^\prime)\delta^{\frac{d-4}{2}}(\sigma) +\beta^\prime_d V(x,x^\prime) \Theta(\sigma) + \beta_d U'(x,x^\prime)\delta^{\frac{d-4}{2}}(\sigma_-) +\beta^\prime_d  V^\prime(x,x^\prime) \Theta(\sigma_-)
		\end{gather}
		where $\sigma_{-}$ is defined as per Equation \eqref{Eq: reflected Synge world function}.
	\end{itemize}
The function $U,U^\prime,V,V^\prime\in C^\infty(\mathcal{O}\times\mathcal{O})$ can be expanded as per Equation \eqref{Eq: expansion of U,V,U',V'} and the coefficients abide by the recursion relations identified by Equations \eqref{Eq: relations for U even dimension}, \eqref{Eq: relations for V even dimension} and \eqref{Eq: recursion for U', V'}.
\end{corollary}

\begin{example}
Having established the counterpart of the Hadamard recursion relations on $\bH^d$ with Robin boundary conditions, it is noteworthy that one can compare them with an exact formula considering the advanced-minus retarded propagator for a massless scalar field in $d=4$, using Equation \eqref{Eq: 4D massless case BF}. In particular this formula entails that
\begin{equation} 
\label{Eq: u_j, v_j 4D massless}
\begin{cases}
        u_0 = u^\prime_0=1, \\
        v_j = 0, \hspace{0.2cm} j \in \mathbb{N}_0, \\
        v'_0 = -\frac{2\kappa}{z+z^\prime}, \\
        v'_j = 2\kappa\sum_{n=0}^{j}A_n(z+z^\prime)\,E_{j-n}(z+z^\prime,\kappa) \hspace{0.2cm} j \in \mathbb{N},
    \end{cases}
\end{equation}
where both $A_n$ and $E_{j-n}$ are defined in Equation \eqref{Eq: 4D massless case BF}. In addition, Equation \eqref{Eq: recursion for U', V'} entails that, for $j=0$, the Hadamard recursion relations yield
\begin{equation}
\label{Eq: recursions v'0 4D}
    \begin{cases}
        \Box u'_0 + 2 \sigma_-^\mu \partial_\mu v'_0+2v'_0 = 0, \\
        \partial_z (u_0 + u'_0) |_{z=0} + \partial_z \sigma (v_0-v'_0) |_{z=0} = \kappa (u_0+u'_0)|_{z=0 }
    \end{cases} .
\end{equation}
Using that $u'_0 = 1$ and $v'_0 = -\frac{2\kappa}{z+z^\prime}$,  
\begin{equation}
        \sigma_-^\mu \partial_\mu v'_0 = \frac{2\kappa}{z+z^\prime} = -v'_0,
\end{equation}
so that the first equation is satisfied. Focusing on the boundary conditions, we have that
\begin{equation}
\partial_z (u_0 + u'_0) |_{z=0} + \partial_z \sigma_- (v_0-v'_0) |_{z=0} = 2\kappa = \kappa (u_0+u'_0)|_{z=0}.
\end{equation}
Hence $v'_0$ satisfies the Hadamard recursion relations with Robin boundary conditions. Focusing on the coefficients $\{v^\prime_j\}_{j \ge 1}$ and bearing in mind that $v_j=0$ for $j \in \mathbb{N}_0$, we obtain from Equation \eqref{Eq: 4D massless case BF} that
\begin{equation} 
\label{Eq: recursions v'j 4D massless}
    \begin{cases}
        \Box v'_j + 2(j+1) \sigma_-^\mu \partial_\mu v'_{j+1}+2(j+1)(j+2)v'_{j+1} = 0, \\
        v^\prime_{j+1}|_{z=0}=\frac{(\partial_z+\kappa)(v^\prime_j)|_{z=0}}{(j+1)z^\prime}, \hspace{0.2cm} j \in \mathbb{N}_0.
    \end{cases}
\end{equation}
These are the Hadamard recursion relations as in Equation \eqref{Eq: recursion for U', V'}.
\end{example}

\subsubsection{Hadamard Recursion Relations: Odd Dimensional Case}
We discuss odd dimensions by focusing only on the main formulae since their derivation is structurally identical to that in \ref{Sec: Hadamard Recurstion Relations even case}. Mimicking the standard construction of the Hadamard recursion relations, see {\it e.g.} \cite{Garabedian_1964, Decanini:2005eg}, we consider the following expansions
\begin{equation} \label{Eq: expansion of U,V,U',V' odd}
	\begin{cases}
		U(x,x^\prime) = \sum_{j=0}^{\infty} u_j (\underline{x},z, \underline{x}',z') \sigma^j \\
		U'(x,x^\prime) = \sum_{j=0}^{\infty} u'_j (\underline{x},z, \underline{x}',z') \sigma_-^j 
	\end{cases}.
\end{equation}
As in the even scenario the coefficients of $U$ abide by the standard Hadamard recursion relations, namely
\begin{equation} \label{Eq: relations for U odd dimension}
	\begin{cases}
		(2-d) \sigma^\mu \partial_\mu u_0 = 0, \\
		[u_0] = 1, \\\\
		Pu_j + (2j+4-d)\sigma^\mu \partial_\mu u_{j+1} + (j+1)(2j+4-d)u_{j+1} = 0, \\
		[u_{j+1}] = -\frac{[Pu_j]}{2j+4-d}, \hspace{0.2cm} j\geq 0 
	\end{cases}
\end{equation}    
Focusing instead on $U^\prime$ we end up with
\begin{equation} 
	\label{Eq: recursion for U', V' odd}
	\begin{cases}
		(2-d) \sigma_-^\mu \partial_\mu u'_0 = 0, \\
		u'_0 |_{z=0} = u_0 |_{z=0}, \\\\
		Pu'_j + (2j+4-d) \sigma_-^\mu \partial_\mu u'_{j+1} + (j+1)(2j+4-d) u'_{j+1} = 0, \\
		\partial_z (u_j+u'_j) |_{z=0} + \frac{1}{2}(2j+4-d)\partial_z \sigma (u_{j+1}-u'_{j+1})|_{z=0} = \kappa (u_j+u'_j)|_{z=0} ,\hspace{0.2cm} j\geq 0.
	\end{cases}
\end{equation}

\begin{remark}
	Observe that, as in the even case Equation \eqref{Eq: recursion for U', V' odd} is guaranteeing that $u_0=u^\prime_0$, which is tantamount to saying that the leading singularity of Equation \eqref{Eq: Ansatz on the Robin parametrix} is the same as that of $\omega_{2,N}$ while boundary conditions contribute to lower order terms.
\end{remark}

To conclude the section we remark that, as in the even case, the recursion relations that we have derived can also be applied at the level of the advanced-minus-retarded propagator as one can infer directly from Equation \eqref{Eq: State boundary condition}. We summarize this statement in the following corollary.

\begin{corollary}\label{Cor: Fundamental Solutions odd}
	Let $G_\kappa\in\mathcal{D}^\prime(\bH^d\times\bH^d)$ be the advanced-minus-retarded propagator for the Klein-Gordon operator with Robin boundary conditions in odd spacetime dimensions $d>1$. Then 
	\begin{itemize}
		\item[\ding{104}] if $\mathcal{O}\cap\partial\bH^d=\emptyset$,  $G_{\kappa}|_{\mathcal{O}\times\mathcal{O}}(x,x^\prime)$, the integral kernel of $G_\kappa$, reads 
		$$G_\kappa(x,x^\prime)= \mathrm{sgn}(\Delta t)\,\frac{(-1)^{\frac{d-1}{2}}}{2\pi^\frac{d}{2}}\frac{U(x,x^\prime)}{(2\sigma)^{\frac{d-2}{2}}(x,x^\prime)}\Theta(\sigma(x,x^\prime)),$$
		\item[\ding{104}] if $\mathcal{O}\cap\partial\bH^d\neq\emptyset$, then
		\begin{gather} 
			G_\kappa(x,x') = \mathrm{sgn}(\Delta t)\,\frac{(-1)^{\frac{d-1}{2}}}{2\pi^\frac{d}{2}}\left(\frac{U(x,x^\prime)}{(2\sigma)^{\frac{d-2}{2}}(x,x^\prime)}\Theta(\sigma(x,x^\prime)) + \frac{U^\prime(x,x^\prime)}{(2\sigma_-)^{\frac{d-2}{2}}(x,x^\prime)}\Theta(\sigma_-(x,x^\prime))\right)
		\end{gather}
		where $\sigma_{-}$ is defined as per Equation \eqref{Eq: reflected Synge world function}.
	\end{itemize}
	The function $U,U^\prime\in C^\infty(\mathcal{O}\times\mathcal{O})$ can be expanded as per Equation \eqref{Eq: expansion of U,V,U',V' odd} and the coefficients abide by the recursion relations identified by Equations \eqref{Eq: relations for U odd dimension} and \eqref{Eq: recursion for U', V' odd}.
\end{corollary}

\subsection{Hadamard states}\label{Sec: Comparison of Hadamard states}

In this last section we discuss how to formulate a counterpart of the celebrated theorems by Radzikowski \cite{Radzikowski_1996, Radzikowski_1996_1} which are applicable when the underlying background is globally hyperbolic.

%
%

\begin{theorem}[Global-to-Local]\label{Thm: Global-to-Local}
Let $(\mathbb{H}^d, \eta)$ be the d-dimensional half-Minkowski spacetime with $d\geq 2$, let $P: = \Box_{\eta} + m^2$ be the massive Klein-Gordon operator thereon and consider $\omega_{2, \kappa} \in  \mathcal{D}'(\mathbb{H}^d \times \mathbb{H}^d)$ abiding by Robin boundary conditions with $\kappa\geq 0$. Given a Cauchy hypersurface $\mathcal{C}$ in $(\mathbb{R}^d, \eta)$ and a causal normal neighbourhood $\mathcal{O} \subseteq \mathbb{R}^d$ we consider the intersection $\tilde{\mathcal{C}} := \mathcal{C} \cap \mathcal{O} \vert_{\mathbb{H}^d}$. The following statements are equivalent: 
\begin{itemize}
    \item[1.] $\omega_{2, \kappa}$ abides by Definition \ref{Def: Local Hadamard State}. 
    \item[2.] $\omega_{2, \kappa}$ satisfies Equation \eqref{Eq: State boundary condition} and the [CCR] in Equation \eqref{Eq: constraints on 2-pt function} modulo $C^{\infty}$, and
    \begin{equation*}
        \text{WF}(\omega_{2, \kappa}) = \{ (x, k, x', -k') \in T^{*}(\mathring{\mathbb{H}}^d \times\mathring{\mathbb{H}}^d) \setminus \{\boldsymbol{0}\} \, | \, (x, k) \, \dot{\sim} \, (x', k') \, \text{and} \, k \rhd 0 \},
    \end{equation*}
    where $\dot{\sim}$ indicates that the two points are connected by a generalized broken bicharacteristic (GBB) as per Definition \ref{Def: generalized broken bicharacteristics}.
\end{itemize}
    
\end{theorem}

\begin{proof}
$1 \Rightarrow 2$.  Following the same rationale as in \cite[Thm. 5.1]{Radzikowski_1996}, we focus on Equation \eqref{Eq: Ansatz on the Robin parametrix} starting from the bidistributions 
\begin{equation}
\label{Eq: b1, b2}
    b_1(x, x') := \lim_{\epsilon \rightarrow 0^+} \frac{1}{\sigma_{\epsilon}^{\frac{d-2}{2}}(x, x^\prime)}, \, \, \, b_2(x, x') := \lim_{\epsilon \rightarrow 0^+} \ln(\sigma_{\epsilon}(x, x^{\prime})), 
\end{equation}
where $\sigma_{\epsilon} (x, x')$ is defined as in Equation \eqref{Eq: local Hadamard form}. In addition, we set  
\begin{equation}
\label{Eq: b1-, b2-}
    b_{1, -}(x, x') := \lim_{\epsilon \rightarrow 0^+} \frac{1}{\sigma^{\frac{d-2}{2}}_{-, \epsilon}(x, x')}, \, \, \, b_{2, -}(x, x') := \lim_{\epsilon \rightarrow 0^+} \ln(\sigma_{-, \epsilon}(x, x')), 
\end{equation}
where $\sigma_{-, \epsilon}(x, x') = \sigma_{-}(x,x') + 2 i \epsilon(t-t^\prime)+ \epsilon^2$, with $\sigma_{-}$ as per Equation \eqref{Eq: reflected Synge world function H^d}. Observe that both kernels in Equation \eqref{Eq: b1, b2} are defined on $\mathcal{O}\times\mathcal{O}\subset\bR^d\times\bR^d$. Therefore we can invoke directly \cite[Thm. 5.1]{Radzikowski_1996} to conclude that 
\begin{equation}
	\label{Eq: WF set b1m}
	\text{WF}(b_1) =\text{WF}(b_2)=\{ (x, k, x', -k') \in T^{*}(\mathcal{O} \times\mathcal{O}) \setminus \{\boldsymbol{0}\} \, | \, (x, k) \, \sim \, (x', k') \, \text{and} \, k \rhd 0 \} ,
\end{equation}
where $\sim$ entails that $x$ and $x^\prime$ are connected by a lightlike geodesic $\gamma$ to which $k$ is co-tangent, while $k^\prime$ is obtained from $k$ by parallel transport. Focusing on $b_{1,-}$ and $b_{2,-}$ we observe that 
$$b_{i,-}=(\iota^*_z\otimes\text{id})b_i, \, \, i = 1, 2,$$
where $\iota_z$ is the reflection map in Definition \ref{Def: Reflected Synge's world function}. Hence, $\iota_z\otimes\text{id}$ is a diffeomorphism of $\bR^d\times\bR^d$, which entails \cite[Chap VIII]{Hormander_1990}
\begin{gather}
	\text{WF}(b_{1,-}) =\text{WF}(b_{2,-})=\text{WF}((\iota^*_z\otimes\text{id})b_i)=\notag\\
	=\{ (\iota_z(x), \iota^*_z(k), x', -k') \in T^{*}(\mathcal{O} \times\mathcal{O}) \setminus \{\boldsymbol{0}\} \, | \, (x, k) \, \dot{\sim} \, (x', k') \, \text{and} \, k \rhd 0 \},	\label{Eq: WF set b1m-}
\end{gather}
where, denoting $k=(k_0,k_1,\dots,k_{d-2},k_z)$, $\iota^*_z(k)=(k_0,k_1,\dots,k_{d-2},-k_z)$. To conclude we observe that $\omega_{2,\kappa}(x,x^\prime)=\Theta(z)\Theta(z^\prime)\tilde{\omega}_{2,\kappa}(x,x^\prime)$ where $\tilde{\omega}_{2,\kappa}$ has the same integral kernel of $\omega_{2,\kappa}$ though without the constraint that $z,z^\prime\geq 0$. On the one hand, we observe that this product of distributions is well defined on account of \cite[Thm. 8.2.10]{Hormander_1990}. On the other hand, it turns out that, for $z,z'>0$, being the Heaviside step functions smooth and equal to $1$, $\text{WF}(\omega_{2,\kappa})=\text{WF}(\tilde{\omega}_{2,\kappa})$. Combining this statement with Equations \eqref{Eq: WF set b1m} and \eqref{Eq: WF set b1m-}, the sought result descends.

\vskip .2cm

\noindent $2 \Longrightarrow 1$. Observe that the two-point correlation function in Proposition \ref{Prop: Hadamard 2-point function with Robin boundary conditions} satisfies the hypothesis on account of Theorem \ref{Thm: State is of local Hadamard form}. Consider a second state $\omega^\prime_{2,\kappa}$ which abides by the hypotheses of point {\em 2.}. Given any convex open subset $\mathcal{O}$ we call $\Delta\omega_{2,\kappa}=\omega^\prime_{2,\kappa}-\omega_{2,\kappa}$ restricted to $\mathcal{O}\times\mathcal{O}$. Since both two-point correlation functions must abide by the CCR condition in Equation \eqref{Eq: constraints on 2-pt function}, working at the level of integral kernels, it holds that $\Delta\omega_{2,\kappa}(x,x^\prime)=\Delta\omega_{2,\kappa}(x^\prime,x)$. At the same time 
$$\text{WF}(\Delta\omega_{2,\kappa})\subset\text{WF}(\omega_{2,\kappa})\cup\text{WF}(\omega^\prime_{2,\kappa})=\text{WF}(\omega_{2,\kappa}).$$ 
Yet, being $\Delta\omega_{2,\kappa}$ symmetric, it means that, if $(x,k,x^\prime,k^\prime)\in\text{WF}(\Delta\omega_{2,\kappa})$, then $(x^\prime,k^\prime,x,k)\in\text{WF}(\Delta\omega_{2,\kappa})$. Per hypothesis $k,-k^\prime\triangleright 0$, which means that $k^\prime$ is past directed. At the same time, if $(x^\prime,k^\prime,x,k)\in\text{WF}(\Delta\omega_{2,\kappa})$, then $k'$ is future directed and, therefore $k=k^\prime=0$, which entails $\text{WF}(\Delta\omega_{2,\kappa})=\emptyset$.

\end{proof}
 
We conclude by commenting on the r\^ole of the Feynman propagator. Following \cite{Duistermaat_1972} and the theory of distinguished parametrices, we can give the following definition.

\begin{definition}\label{Def: Feynman Parametrix}
	Let $\bH^d$ be the $d$-dimensional half-Minkowski spacetime as per Equation \eqref{Eq: half Minkowski} and let $P$ be the Klein-Gordon operator. We say that $H^F_\kappa\in\mathcal{D}^\prime(\bH^d\times\bH^d)$ is a {\bf Robin-Feynmann parametrix} if 
	\begin{enumerate}
		\item it holds that
		$$(P\otimes\mathbb{I})H^F_\kappa=\delta+R_1\quad\textrm{and}\quad(\mathbb{I}\otimes P)H^F_\kappa=\delta+R_2,$$
		where $R_1,R_2\in C^\infty(\bH^d\times\bH^d)$. 
		\item given $\kappa>0$
		$$(\partial_{\mathbf{n}} H^F_\kappa + \kappa H^F_\kappa) \vert_{\partial \mathcal{M}} = 0.$$
		\item the singular structure of $H_{F,\kappa}$ is encoded by
		\begin{equation}\label{Eq: WF Robin Feynmann parametrix}
				WF(H^F_\kappa) = \{(x,k_x, x', k_{x'}) \in T^*(\mathring{\bH}^d\times\mathring{\bH}^d)\setminus\{0\} \, | \, (x,k_x) \, \dot{\sim}_{F} \, (x', -k_{x'}), k_x \ne 0 \} \cup \text{WF}(\delta_{2}),
		\end{equation}
	where $\dot{\sim}_{F}$ entails that there exists a generalized broken bicharacteristic connecting $(x,k_x)$ to $(x', -k_{x'})$ such that $-k_{x'}$ is the parallel transport of $k_x$ along it and if $x' \in J^{+}(x)$, then $k_x \triangleright 0$ (\emph{resp}. $k_x \triangleleft 0$), while if $x' \in J^{-}(x)$, then $k_x \triangleleft 0$ (\emph{resp}. $k_x \triangleright 0$).
\end{enumerate}
\end{definition}

\begin{proposition}\label{Prop: Existence of Feynman parametrices}
	Under the assumptions of Definition \ref{Def: Feynman Parametrix}, given $H_\kappa$ a Robin parametrix as per Definition \ref{Def: Robin-Hadamard paraemtrix} and $G^-_\kappa$ as per Equation \eqref{Eq: Robin fundamental solution via Fulling}, then $G^-_\kappa-iH_{\kappa}$ is a Feynman parametrix. 
\end{proposition}

\begin{proof}
	On account of Theorem \ref{Thm: Global-to-Local} it suffices to prove the statement working with $\omega_{2,\kappa}$ as per Equation \eqref{Eq: omega2 Robin}. One can infer that, denoting by $G^F_\kappa\doteq G^-_\kappa-i\omega_{2,\kappa}$, it holds per construction that $(P\otimes\mathbb{I})G^F_\kappa=(\mathbb{I}\otimes P)G^F_\kappa=\delta$ and it abides by the Robin boundary conditions inheriting them from $G^-_\kappa$ and $\omega_{2, \kappa}$. In order to control the underlying wavefront set observe that, at the level of integral kernels, $G^F_\kappa(x,x^\prime)=\Theta(t-t^\prime)\omega_{2, \kappa}(x,x^\prime)+\Theta(t'-t)\omega_{2, \kappa}(x',x)$. Using Equation \eqref{Eq: WF Robin Feynmann parametrix}, one can apply the same line of reasoning used on globally hyperbolic spacetimes to draw the sought conclusion.
\end{proof}

\section*{Acknowledgements}
\addcontentsline{toc}{section}{Acknowledgements}

We are grateful to Nicola Pinamonti for enlightening discussions. The work of BC has been supported by a fellowship of the University of Pavia and BC is grateful to the RQI Cost Action CA23115 and to the Department of Mathematics at the University of York for the kind hospitality during the realization of part of this work. BC and CD both acknowledge the support of the INFN Sezione di Pavia and of Gruppo Nazionale di Fisica Matematica, part of INdAM. The work of BC is supported in part by a fellowship of the "Progetto Giovani GNFM 2025" under the project "Hadamard states for linearized Yang-Mills theories" fostered by Gruppo Nazionale di Fisica Matematica -- INdAM in collaboration with Simone Murro (University of Genova). BAJ-A is supported by the EPSRC Open Fellowship EP/Y014510/1.  Part of this work has appeared in the Master thesis of RDS submitted on the 19/09/2025 as a partial fulfillment of the requirements to obtain a Master degree in Physics at the University of Pavia. 

\vskip.2cm

\noindent\textbf{Data availability statement}. Data sharing is not applicable to this article as no new data were created or analysed in this study.

\vskip .2cm

\noindent\textbf{Conflict of interest statement.} The authors certify that they have no affiliations with or involvement in any
organization or entity with any financial interest or non-financial interest in the subject matter discussed in
this manuscript.

\appendix

\section{Generalized Broken Bicharacteristics}\label{Sec: GBB}

In this appendix we discuss some distinguished geometric features of a class of Lorentzian manifolds which have half-Minkowski spacetime as per Equation \eqref{Eq: half Minkowski} as their prototype. More precisely, we call {\em spacetime}, $(\mathcal{M},g)$, a $d-$dimensional, oriented and connected manifold, $d \ge 2$, endowed with a smooth Lorentzian metric $g$ of signature $(d-1,1)$, \cite[Ch. 2]{Lee_2018}. In addition, we allow $(\mathcal{M}, g)$ to have possibly a {\em non-empty boundary}, $\partial \mathcal{M}$, and we require that it is also {\em time-oriented}, namely there exists a timelike vector field $\chi\in\Gamma(T\mcM)$ which has been chosen once and for all. 

\begin{definition}
\label{Def: time-like boundary}
Under the previous assumptions, we call $\partial \mathcal{M}$ a {\bf timelike boundary} if, denoting by $\iota: \partial \mathcal{M} \hookrightarrow \mathcal{M}$ the natural inclusion map, the pullback $\iota^* g$ identifies a Lorentzian metric on $\partial \mathcal{M}$. In this case, $(\mathcal{M}, g)$ is said to be a \emph{spacetime with timelike boundary}. 
\end{definition}

\noindent Among the set of oriented and time-oriented Lorentzian manifolds with $\partial \mcM=\emptyset$, a distinguished class is that of globally hyperbolic spacetimes, since they possess a non pathological causal structure and, in addition, one can formulate a well-posed initial value problem for second order, partial differential, scalar equations ruled by symmetric hyperbolic operators, see {\it e.g} \cite{Baer_2007}. This notion has been originally formulated under the requirement that the underlying manifold has an empty boundary, but a generalization which accounts also for those spacetimes abiding by Definition \ref{Def: time-like boundary} has been discussed in \cite{Ak_Hau_2020}.

\begin{definition} \label{Def: Globally Hyperbolic}
\label{Def: ghm}
A spacetime $(\mathcal{M},g)$ with a timelike boundary as per Definition \ref{Def: time-like boundary} is said to be {\bf globally hyperbolic} if the following two conditions are met: 
\begin{itemize}
    \item[(i)] $(\mathcal{M
    }, g)$ is strongly causal, \textit{i.e.}, it does not contain any closed, causal curve. 
    \item[(ii)] for any pair of points $p, q \in \mathcal{M}$, denoting by $J^{\pm}$ the causal future and past respectively, the causal diamond $J^+(p) \cap J^-(q)$ is a compact set.
\end{itemize}
\end{definition}

\noindent A notable and very useful property of globally hyperbolic spacetimes has been proven in \cite[Thm 1.1]{Ak_Hau_2020} as a direct consequence of $(\mathcal{M},g)$ possessing a Cauchy orthogonal splitting. Here we report part of this result, limited to those aspects which will be of interest in our investigation.
\begin{proposition}\label{Prop: Globally Hyperbolic}
Let $(\mcM,g)$ be a globally hyperbolic spacetime as per Definition \ref{Def: Globally Hyperbolic}. There exists an isometry $\psi: \mathbb{R} \times \Sigma \rightarrow \mathcal{M}$ such that the line element associated to $\psi^* g$ is of the form
\begin{equation}
	\label{Eq: line element psi*g}
	ds^2 = \beta dt^2 - h_t,
\end{equation}
where $\{t\}\times\Sigma$ is a smooth, spacelike, Cauchy hypersurface for all $t\in\bR$, while $t: \mathbb{R} \times \Sigma \rightarrow \mathbb{R}$ is a temporal function whose gradient $\nabla_t$ is tangent to $\partial\mcM$. In addition $\beta \in C^{\infty}(\mathbb{R} \times \Sigma;(0,\infty))$, while $\{h_t\}_{t \in \mathbb{R}}$ is a family of smooth Riemannian metrics on $\{t\} \times \Sigma$, varying smoothly with $t$.
\end{proposition}

Henceforth, we shall work on a globally hyperbolic spacetime $(\mcM,g)$, $\text{dim} \, \mathcal{M} = d \ge 2$, with a timelike boundary. Furthermore, we assume that the reader is already familiar with $b$-geometry techniques, see \cite{Melrose_1993} for a detailed account. In the same spirit of \cite{GMP}, we will content ourselves with setting the basic notations and conventions. 

Under the aforementioned assumptions, let us denote by ${}^bT\mcM$ the {\em $b$-tangent bundle}, namely, a vector bundle whose fibres are
$${}^bT_p\mcM=\left\{ 
\begin{array}{ll}
	T_p \mcM & p\in \mathring{M}\\
	\textrm{span}_{\mathbb{R}}(z\partial_z, T_p\partial M) & p \in \partial M
\end{array}
\right., $$
where $z:\mcM\to\bR$ is a global boundary function, \emph{i.e.}, $\partial\mcM$ corresponds to $z=0$, which is here promoted to coordinate. Relying on a duality argument, we can also introduce the {\em $b$-cotangent bundle}, ${}^bT^*\mcM$, \textit{i.e.}, a vector bundle whose fibers are
$${}^bT^*_p\mcM=\left\{ 
\begin{array}{ll}
	T^*_p \mcM & p\in \mathring{\mcM}\\
	\textrm{span}_{\mathbb{R}}(\frac{dz}{z}, T^*_p\partial \mcM) & p \in \partial\mcM
\end{array}. 
\right.$$
For later convenience, in any but fixed chart $U$ of $\mcM$ centered at a point $p\in\partial\mcM$, we will denote by $\{(z,x_i,\xi,\eta_i)\}_{i=1}^{d-1}$ and by $\{(z,x_i,\zeta,\eta_i)\}_{i=1}^{d-1}$, with $d-1=\dim\partial\mcM$, local coordinates respectively of $T^*\mcM|_U$ and of ${}^bT^*\mcM|_U$. Since $(\mathcal{M},g)$ is globally hyperbolic as per Definition \ref{Def: Globally Hyperbolic}, we can identify implicitly $\eta_{d-1}\equiv t$, where $t$ is the distinguished time direction on $\mathcal{M}$, {\it cf.} Equation \eqref{Eq: line element psi*g}. Note that there exists also a natural projection map 
\begin{equation}
    \label{Eq: projection map}
    \pi:T^*\mcM\to{}^b T^*\mcM,\quad (z,x_i,\xi,\eta_i)\mapsto \pi(z,x_i,\xi,\eta_i)=(z,x_i,z\xi,\eta_i),
\end{equation}
which, however, fails to be injective. The projection map $\pi$ allows for the construction of a key geometric structure, \emph{i.e.}, the {\em compressed $b$-cotangent bundle} 
\begin{equation}\label{Eq: Compressed b-cotangent bundle}
	{}^b\dot{T}^*\mcM\doteq\pi[T^*\mcM],
\end{equation}
which is a vector sub-bundle of ${}^bT^*\mcM$, such that, for any $p\in\mathring{\mcM}$, ${}^b\dot{T}^*_p\mcM\equiv T^*_p\mcM$. In addition, we define the {\em b-cosphere bundle} as the quotient manifold obtained via the action of the dilation group on $T^*_b\mcM\setminus\{0\}$, that is
\begin{equation}\label{Eq: cosphere bundle}
	{}^bS^*\mcM\doteq\bigslant{{}^bT^*\mcM\setminus\{0\}}{\mathbb{R}^+}.
\end{equation}
Observe that, taking any local chart $U\subset M$ such that $U\cap\partial\mcM\neq\emptyset$ and the associated local coordinates $\{(z,x_i,\zeta,\eta_i)\}_{i=1}^{d-1}$, with $d-1=\dim\partial\mcM$, on ${}^bT^*_U \mcM\doteq{}^bT^*\mcM|_U$, it is possible to devise a natural counterpart on ${}^bS^*\mcM \vert_{U}$, namely, $\{(z,x_i,\widehat{\zeta},\widehat{\eta}_i)\}_{i=1}^{d-1}$ where $\widehat{\zeta}=\frac{\zeta}{|\eta_{d-1}|}$ and $\widehat{\eta_i}=\frac{\eta_i}{|\eta_{d-1}|}$, $i=1,..., d-1$.

If we confine our attention to the Klein-Gordon operator $P$ as per Equation \eqref{Eq: KG equation}, whose principal symbol reads $\widehat{p}\doteq\widehat{g}(X,X)$, where $X\in\Gamma(T^*\mcM)$, the associated {\em characteristic set} is of the form
\begin{equation}\label{Eq: characteristic set}
	\mathcal{N}=\left\{(q,k_q)\in T^*\mcM\setminus\{0\}\;|\; \widehat{g}^{ij}(k_q)_i (k_q)_j=0\right\},
\end{equation}
whilst we define the {\em compressed characteristic set} as 
\begin{equation}\label{Eq: compressed characteristic set}
	\dot{\mathcal{N}}=\pi[\mathcal{N}]\subset{}^b\dot{T}^*\mcM,
\end{equation}
$\pi$ being the projection map in Equation \eqref{Eq: projection map}. $\dot{\mathcal{N}}$ can be equipped with the subspace topology inherited from the compressed b-cotangent bundle ${}^b\dot{T}^*\mcM$. In addition, it is convenient to identify the following three conic subsets of ${}^b\dot{T}^*\mcM$:

\vskip.2cm

\begin{itemize}
	\item[\ding{104}] The {\em elliptic} region 
	\begin{equation}\label{Eq: elliptic region}
		\mathcal{E}(\mcM) = \{ q \in{}^b\dot{T}^*\mcM \setminus\{0\} \ : \ \pi^{-1}(q) \cap \mathcal{N} = \emptyset \},
	\end{equation}
	where $\pi:T^*\mcM\to{}^b\dot{T}^*\mcM$.
	
	\vskip.2cm
	
	\item[\ding{104}] The {\em glancing} region  
	\begin{equation}\label{Eq: glancing region}
		\mathcal{G}(\mcM) = \{ q \in{}^b\dot{T}^*\mcM \setminus\{0\} \ : \ Card( \pi^{-1}(q) \cap \mathcal{N}) = 1 \},
	\end{equation}
	where $Card$ is the cardinality of a set.
	
	\vskip .2cm
	
	\item[\ding{104}] The {\em hyperbolic} region 
	\begin{equation}\label{Eq: hyperbolic region}
		\mathcal{H}(\mcM) = \{ q \in{}^b\dot{T}^*\mcM \setminus\{0\} \ : \ Card( \pi^{-1}(q) \cap \mathcal{N}) = 2 \}.
	\end{equation}
\end{itemize}

A related concept is the following: 

\begin{definition}\label{Def: generalized broken bicharacteristics}
Let $I \subset \mathbb{R}$, we say that a continuous map $\gamma : I \rightarrow \dot{\mathcal{N}}$ is a {\bf generalized broken bicharacteristic} (GBB) if, for any $s_0 \in I$, the following conditions hold true:
	\begin{itemize}
		\item[a)] Denoting by $\pi:T^*\mcM\to{}^bT^*\mcM$ the projection map onto the b-cotangent bundle as per Equation \eqref{Eq: projection map}, if $q_0 = \gamma(s_0) \in \mathcal{G}(\mcM)$, then for every $\omega\in \Gamma^\infty(^bT^*M)$, 
		\begin{equation}
			\frac{d}{ds}(\omega \circ \gamma) = \{ \widehat{p},\pi^* \omega \}(\eta_0),
		\end{equation}
		where $\eta_0 \in \mathcal{N}$ is the unique point for which $\pi(\eta_0)=q_0$, while $\{,\}$ are the Poisson brackets on $T^*\mcM$.
		\item[b)] If $q_0 = \gamma(s_0)$ lies in $ \mathcal{H}(\mcM)$, then $\exists \varepsilon > 0$ such that $z(\gamma(s))\neq 0$, whenever $0 < |s-s_0| < \varepsilon $,  where $z$ is the global boundary function.
	\end{itemize}
\end{definition}

\end{document}